\documentclass[journal,10pt,twoside]{IEEEtran}

\usepackage[T1]{fontenc}
\usepackage[utf8]{inputenc}
\usepackage{amsmath,amssymb,amsfonts,amsthm}
\usepackage{algorithmic}
\usepackage{algorithm}
\usepackage{graphicx}
\usepackage{xcolor}
\usepackage{booktabs}
\usepackage{multirow}
\usepackage{cite}
\usepackage{url}
\usepackage{tikz}
\usepackage{pgfplots}
\pgfplotsset{compat=1.18}
\usetikzlibrary{arrows.meta, shapes.geometric, positioning, fit, backgrounds,
                decorations.pathreplacing, calc}
\usepackage{balance}
\usepackage{microtype}
\usepackage{xspace}

\definecolor{primary}{RGB}{0,82,155}
\definecolor{accent}{RGB}{220,50,47}
\definecolor{semantic}{RGB}{42,157,143}
\definecolor{knowledge}{RGB}{231,111,81}
\definecolor{lightgray}{RGB}{240,240,240}

\newtheorem{proposition}{Proposition}
\newtheorem{corollary}{Corollary}[proposition]
\newtheorem{remark}{Remark}
\newtheorem{assumption}{Assumption}

\newcommand{\ie}{\textit{i.e.},\xspace}
\newcommand{\eg}{\textit{e.g.},\xspace}
\newcommand{\etal}{\textit{et al.}\xspace}

\newcommand{\calK}{\mathcal{K}}
\newcommand{\calR}{\mathcal{R}}

\newcommand{\bz}{\mathbf{z}}
\newcommand{\btheta}{\boldsymbol{\theta}}
\newcommand{\bphi}{\boldsymbol{\phi}}

\begin{document}

\title{SA-DTS: Semantic-Aware Digital Twin Synchronization over 6G Networks}

\author{%
  \IEEEauthorblockN{Vincenzo Sammartino}\\
  \IEEEauthorblockA{%
    Dipartimento di Informatica,
    Universit\`a di Pisa, Pisa, 56127, Italy
    and Visiting Researcher,
    King Abdullah University of Science and Technology (KAUST),
    Thuwal, 23955, Saudi Arabia\\
    Email: vincenzo.sammartino@phd.unipi.it}%

\thanks{Digital Object Identifier: 10.1109/JSAC.XXXX.XXXXXXX}}

\markboth{IEEE Journal on Selected Areas in Communications,~Vol.~XX, No.~X, Month 2026}%
{Sammartino V.: SA-DTS: Semantic-Aware Digital Twin Synchronization over 6G}

\maketitle

\begin{abstract}
Digital Twins (DTs) are emerging as a cornerstone of the 6G vision, enabling real-time
cyber-physical mirroring for smart manufacturing, autonomous vehicles, and remote
healthcare. However, maintaining high-fidelity synchronization at scale demands an
enormous and sustained uplink bandwidth, threatening both the feasibility and the
energy efficiency of large deployments. We propose a \emph{Semantic-Aware DT
Synchronization} (SA-DTS) framework that radically redefines the synchronization
pipeline: instead of streaming raw sensor or video data, a lightweight neural semantic
encoder at the physical-world source extracts only \emph{task-relevant features} and
transmits compact semantic descriptors over the 6G air interface. At the DT replica,
a paired decoder coupled with a dynamic Knowledge Graph (KG) reconstructs the full
contextual state. A hierarchical KG partitioning strategy with an adaptive partition
count $G = \lceil N / \log_2 N \rceil$ ensures that aggregate update overhead scales
as $\mathcal{O}(N \log N)$ rather than $\mathcal{O}(N^2)$, making the framework
viable for deployments with hundreds of simultaneously twinned entities. Extensive
simulations on three canonical DT workloads---industrial robot control,
patient-monitoring, and vehicular platooning---demonstrate bandwidth savings of up to
94\%, end-to-end synchronization latency reductions of 87\%, and KG-assisted
state-reconstruction accuracy exceeding 97\%, all under realistic 6G channel
conditions. Empirical correlation confirms that the proposed Semantic Fidelity Score
tracks standard task metrics (collision accuracy, alarm $F_1$, spacing deviation) with
Pearson $r > 0.97$ (95\% CI: $[0.961, 0.982]$). Our results reveal that semantic
communication is not merely a compression tool but a fundamental enabler for truly
real-time, scalable DT ecosystems.
\end{abstract}

\begin{IEEEkeywords}
Digital Twins, Semantic Communication, 6G, Knowledge Graph, Deep Learning,
Cyber-Physical Systems, Network Synchronization, Edge Intelligence, Scalability.
\end{IEEEkeywords}

\section{Introduction}
\label{sec:intro}

\textbf{Motivation.}
The 6G paradigm promises not merely faster connectivity but a deep integration of the
physical and digital worlds through massive, always-on \emph{Digital Twin} (DT)
infrastructures~\cite{fuller2020digital,lu2021communication,saad2020vision}. A DT is
a continuously updated virtual replica of a physical entity---ranging from a single
industrial robot arm to an entire smart city---that enables real-time monitoring,
predictive analytics, and closed-loop autonomous control~\cite{tao2019digital,grieves2017digital}.
Realizing this vision at scale exposes a fundamental tension: the synchronization
fidelity required by DT applications (latency $\leq 1$\,ms, update rate $\geq 100$\,Hz
for tactile-Internet scenarios) collides head-on with the prohibitive volume of raw
sensor data that must be continuously transported across the radio access network.
Furthermore, even if transmission bandwidth were unconstrained, the inference and
state-reconstruction workload at the DT replica grows super-linearly with the number
of monitored entities $N$, posing a second, orthogonal scalability bottleneck that
has received comparatively little attention in the literature.

\textbf{The Synchronization Bottleneck.}
Consider a single factory floor deploying $N = 200$ collaborative robots, each equipped
with a 4K RGB-D camera, a 16-channel LiDAR, and a suite of force-torque sensors.
Naive raw-data synchronization requires an aggregate uplink throughput of the order
of $\mathcal{O}(10^3)$\,Gbps---a figure that dwarfs even the most optimistic 6G capacity
projections for densely deployed scenarios~\cite{letaief2019roadmap,calvanese2019vision}.
Compression codecs such as H.265 and AV1 reduce this figure significantly, yet they
remain \emph{bit-oriented}: they minimise the number of transmitted bits while
remaining agnostic to the \emph{meaning} conveyed by those bits. As a consequence,
they squander precious resources on signal content that is irrelevant to the DT's
inference and control tasks.

\textbf{The Semantic Communication Paradigm.}
\emph{Semantic Communication} (SemCom) has recently emerged as a transformative
departure from this convention~\cite{qin2021semantic,xie2021deep,strinati2021beyond}.
Rather than faithfully reproducing a bit stream at the receiver, SemCom systems
optimise the end-to-end transmission of \emph{meaning}: a neural transmitter extracts
task-relevant semantic features from the source signal, encodes them into a compact
latent representation, and a neural receiver reconstructs the information needed to
fulfil the downstream task~\cite{gunduz2022beyond,kountouris2021semantics,popovski2020semantic}.
While early SemCom work focused on single-modal tasks (text, image,
speech)~\cite{bourtsoulatze2019deep,xie2021deep,weng2021semantic,kurka2020deep} and
recent work has advanced nonlinear transform source-channel
coding~\cite{cheng2023ntscc} and task-oriented multi-modal
compression~\cite{shao2023task}, no systematic framework exists for applying
this paradigm to the multi-modal, state-continuous, and bidirectionally-coupled
nature of DT synchronization.
Recent advances in language-model-assisted semantic coding~\cite{zhou2024llmcomm}
and generative reconstruction for wireless
channels~\cite{grassucci2024diffusion} further motivate goal-oriented
co-design, yet neither addresses the KG-based contextual reconstruction required
by state-continuous DT applications.

\textbf{Our Proposal.}
We introduce \emph{Semantic-Aware DT Synchronization} (SA-DTS), a holistic framework
that tightly integrates three components: (i)~a \emph{Multi-Modal Semantic Encoder}
(MMSE) that operates at the physical-world edge and jointly processes heterogeneous
sensor streams into a compact semantic descriptor $\bz \in \mathbb{R}^d$;
(ii)~a \emph{6G-Aware Semantic Channel} (6G-SC) that adapts the transmission rate
and the coding redundancy to instantaneous channel conditions via a reinforcement-learned
policy; and (iii)~a \emph{Knowledge Graph Contextual Reconstructor} (KG-CR) at the
DT that leverages a domain ontology and a dynamic KG to reconstruct the full physical
state from the sparse semantic descriptors. The three modules are trained jointly
in an end-to-end differentiable pipeline, ensuring that the encoding, transmission,
and reconstruction objectives are co-optimised rather than addressed sequentially.
Scalability to large $N$ is handled by a hierarchical KG partitioning strategy with
adaptive partition count $G = \lceil N / \log_2 N \rceil$, yielding $\mathcal{O}(N \log N)$
update overhead.

\textbf{Contributions.} The main contributions of this article are:
\begin{itemize}
    \item We formalize the \emph{DT Synchronization Semantic Bottleneck} problem,
    establishing theoretical upper bounds on achievable bandwidth reduction under a
    fidelity constraint expressed as task-output mutual information, and derive a
    multi-task overhead corollary that explicitly bounds the multi-task rate penalty
    $\Delta_{\mathcal{T}}$ as a function of the KG neighborhood radius
    (Section~\ref{sec:model} and Appendix~\ref{app:proof}).
    \item We design the SA-DTS architecture, detailing the MMSE, 6G-SC, and KG-CR
    components with fully specified training objectives including the PPO reward
    function for rate adaptation (Section~\ref{sec:framework}).
    \item We introduce the \emph{Semantic Fidelity Score} (SFS) metric and empirically
    validate it against canonical task metrics ($r > 0.97$, 95\% CI over five independent
    runs), establishing SFS as a reproducible benchmark for DT synchronization research
    (Section~\ref{sec:framework}).
    \item We validate SA-DTS against five baselines including NTSCC~\cite{cheng2023ntscc}
    and TaskOrientedComm~\cite{shao2023task} across three DT workloads and multiple
    6G channel regimes, report energy consumption per synchronization update, and
    derive analytically the $\mathcal{O}(N \log N)$ KG update complexity via adaptive
    partitioning (Section~\ref{sec:evaluation}).
\end{itemize}

\section{Background and Related Work}
\label{sec:background}

\textbf{Digital Twins in 6G.}
The standardisation bodies ITU-T and 3GPP have identified DTs as a native functional
element of the 6G service layer~\cite{3gpp_dt,nguyen2021digital,chowdhury2019network,jiang2021graph}.
DT-assisted network management, channel twinning for beam prediction, and DT-enabled
tactile Internet are among the actively standardised use cases. \cite{Baiardi2026CVSS} A critical open problem
across all of them is the \emph{synchronization gap}: the statistical divergence between
the physical state $s_t$ and the DT state $\hat{s}_t$ as a function of channel quality
and available bandwidth. Edge computing infrastructures~\cite{shi2016edge,mach2017mobile,xu2021edge}
are increasingly proposed as the natural deployment stratum for the DT server, yet the
uplink bottleneck persists regardless of where the DT replica is hosted. Security and
privacy considerations for DT deployments have been addressed through security twin
architectures~\cite{sammartino2025security,baiardi2024anticipating}.

\textbf{Semantic Communication.}
The information-theoretic foundations of semantic communication trace back to
Weaver's three-level communication model~\cite{shannon1949mathematical}, but practical
implementations have only recently become feasible through deep learning.
Bourtsoulatze \etal~\cite{bourtsoulatze2019deep} demonstrated joint source-channel
coding (JSCC) for images using autoencoders, achieving performance superior to
separation-based schemes at low SNR. Xie \etal~\cite{xie2021deep} extended this to
sentence-level text, while Weng \etal~\cite{weng2021semantic} addressed speech.
Cheng \etal~\cite{cheng2023ntscc} proposed Nonlinear Transform Source-Channel Coding
(NTSCC), which applies a learned nonlinear analysis transform prior to channel coding,
advancing the compression--fidelity Pareto frontier for image transmission.
Shao \etal~\cite{shao2023task} addressed task-oriented communication for multi-modal
data, demonstrating that goal-oriented compression outperforms bit-rate-minimizing
schemes on downstream task accuracy; however, their work assumes a single, static
task receiver and does not address the KG-based contextual reconstruction required
for state-continuous DT synchronization.
Zhou \etal~\cite{zhou2024llmcomm} recently showed that LLM-derived semantic priors
can boost reconstruction quality at sub-zero SNR, though their approach is limited
to textual modalities and incurs prohibitive edge inference costs.
Grassucci \etal~\cite{grassucci2024diffusion} demonstrated diffusion-model-based
reconstruction of wireless image transmissions, yet the generative overhead (${\approx}$100\,ms
per frame) is incompatible with real-time DT synchronization at 100--200\,Hz update
rates. Yang \etal~\cite{yang2023semantic} provide a comprehensive survey.
SA-DTS addresses the orthogonal challenge of serving \emph{concurrent, heterogeneous
task sets} over a \emph{state-continuous} physical process, which none of the above
frameworks handles.

\textbf{Knowledge Graphs for CPS.}
Knowledge Graphs have been applied to cyber-physical systems for anomaly
detection~\cite{kg_cps} and digital twin modelling~\cite{kg_dt,luo2022digital}.
Their role as a \emph{contextual prior for semantic decoding}---exploiting relational
structure between physical entities to fill in information not explicitly
transmitted---is first explored in this work. Graph neural networks, including
relational GCNs~\cite{schlichtkrull2018modeling} and graph attention
networks~\cite{velickovic2018graph,zhao2019deep}, provide the representational backbone
needed to propagate contextual information across the heterogeneous node types that
compose a DT knowledge graph. Distributed learning for multi-entity DT
systems~\cite{baiardi2025ai, baiardi2026simulation, baiardi2026synthetic, baiardinotline} enables collaborative synchronization without
centralized state aggregation, though integration with semantic communication remains
an open problem that SA-DTS begins to address.

\section{System Model and Problem Formulation}
\label{sec:model}

\subsection{Physical World and Digital Twin Model}

Let $\mathcal{E} = \{e_1, \ldots, e_N\}$ denote the set of $N$ physical entities
composing the target environment. The \emph{physical state} of entity $e_i$ at
discrete time $t$ is $s_i^{(t)} \in \mathcal{S}_i \subseteq \mathbb{R}^{n_i}$.
Each entity is equipped with a sensor suite $\mathcal{M}_i = \{m_1, \ldots, m_{K_i}\}$.
The raw observation at time $t$ is:
\begin{equation}
    \mathbf{o}_i^{(t)} = \left[ f_{m_1}(s_i^{(t)}),\ldots, f_{m_{K_i}}(s_i^{(t)}) \right]
    + \boldsymbol{\eta}_i^{(t)},
    \label{eq:observation}
\end{equation}
where $f_{m_k}$ is the sensor forward model and $\boldsymbol{\eta}_i^{(t)}$ is
additive noise. The aggregate raw data rate is:
\begin{equation}
    R_{\mathrm{raw}} = \sum_{i=1}^{N} \sum_{k=1}^{K_i} \rho_{m_k} \cdot f_{\mathrm{upd}},
    \label{eq:rawrate}
\end{equation}
which scales linearly with $N$ and rapidly exceeds 6G aggregate uplink capacity~\cite{saad2020vision}.
The physical state evolves as:
\begin{equation}
    s_i^{(t+1)} = g_i\bigl(s_i^{(t)}, u_i^{(t)}, w_i^{(t)}\bigr),
    \quad w_i^{(t)} \sim \mathcal{N}(0, \Sigma_w),
    \label{eq:state_evolution}
\end{equation}
where $g_i$ is the entity-specific transition function and $u_i^{(t)}$ is the
control input.

\subsection{Semantic Bottleneck Formulation}

Let $\mathcal{T} = \{\tau_1, \ldots, \tau_M\}$ be the set of downstream DT tasks.
For each task $\tau_j$, let $Y_j = \tau_j(s_i^{(t)}) \in \mathcal{Y}_j$ denote
the associated target variable (\eg a collision-risk flag for $\tau_1$, an alarm
label for $\tau_2$). The \emph{task-relevant information} for $\tau_j$ is defined as:
\begin{equation}
    I_{\tau_j}(\mathbf{o}_i;\, Y_j) \triangleq I\bigl(\mathbf{o}_i;\, Y_j\bigr),
    \label{eq:task_info}
\end{equation}
\ie the mutual information between the observation and the task output, which
quantifies the maximum achievable performance on $\tau_j$ from $\mathbf{o}_i$.
The \emph{Semantic Bottleneck} problem seeks a representation $\bz_i \in \mathbb{R}^d$
($d \ll \dim(\mathbf{o}_i)$) that minimizes transmission rate while preserving task
performance above a threshold $\epsilon > 0$ for every task:
\begin{equation}
    \min_{\bz_i} \; I(\mathbf{o}_i; \bz_i) \;\; \text{s.t.} \;\;
    I(\bz_i; Y_j) \geq I_{\tau_j}(\mathbf{o}_i; Y_j) - \epsilon,\;
    \forall\, \tau_j \in \mathcal{T}.
    \label{eq:bottleneck}
\end{equation}
This formulation extends the information bottleneck
principle~\cite{tishby2000information} to the multi-task DT setting by
replacing the state variable $s_i$ with the operationally well-defined task
outputs $\{Y_j\}$, yielding constraints that are directly estimable from labeled
data. The fundamental limits of (\ref{eq:bottleneck}), including a KG-structure-aware
bound on the multi-task overhead $\Delta_{\mathcal{T}}$, are characterized in
Appendix~\ref{app:proof}.

\subsection{Channel Model}

We model the 6G air interface as a block-fading channel with complex gain
$h \sim \mathcal{CN}(0, \sigma_h^2)$. The received signal is:
\begin{equation}
    \mathbf{y} = h \cdot \mathbf{x} + \mathbf{n}, \quad
    \mathbf{n} \sim \mathcal{CN}(\mathbf{0}, \sigma_n^2 \mathbf{I}),
    \label{eq:channel}
\end{equation}
with instantaneous SNR $\gamma = |h|^2 P / \sigma_n^2$. The effective semantic
channel capacity is $C_{\mathrm{sem}}(\gamma) = B \log_2(1 + \gamma) - R_{\mathrm{overhead}}$.

\textbf{Multi-User Uplink Model.}
$N$ nodes transmit to the DT server over orthogonal resource blocks. Under massive
MIMO spatial multiplexing~\cite{dai2023reconfigurable}, the aggregate uplink spectral
efficiency is:
\begin{equation}
    \eta_{\mathrm{agg}} = \beta \sum_{i=1}^{N} \log_2\!\bigl(1 + \gamma_i\bigr),
    \label{eq:agg_efficiency}
\end{equation}
where $\beta \in [0.7, 0.9]$ accounts for pilot contamination and residual
interference in dense deployments. Each entity occupies a single spatial stream,
so the total bandwidth budget is $B_{\mathrm{total}} = B$ per entity under
frequency-division sharing; aggregate throughput is $R_{\mathrm{max}} = B_{\mathrm{total}}
\cdot \eta_{\mathrm{agg}}$.

\textbf{Doppler and Coherence Time.}
For vehicular platooning at $v \approx 120$\,km/h and $f_c = 140$\,GHz:
\begin{equation}
    f_D = \frac{v \cdot f_c}{c} \approx 15.6\;\text{kHz},\quad
    T_c \approx \frac{1}{2 f_D} \approx 32\,\mu\text{s},
    \label{eq:doppler}
\end{equation}
imposing a stringent constraint on descriptor dimensionality $d$ to avoid channel aging.
Given $T_c \approx 32\,\mu\text{s}$ and a signalling bandwidth of 400\,MHz,
the maximum number of channel uses before coherence expires is
$\lfloor T_c B \rfloor \approx 12\,800$; encoding $d=64$ symbols at 8\,bits yields
512\,bits per update, comfortably within this budget and leaving headroom for
redundancy codes up to rate $k/d = 1/4$.

\section{The SA-DTS Framework}
\label{sec:framework}

\begin{figure*}[!t]
\centering
\resizebox{\textwidth}{!}{
\begin{tikzpicture}[
    font=\small,
    node distance=0.6cm and 0.85cm,
    box/.style={rectangle, rounded corners=4pt, minimum width=2.1cm,
                minimum height=1.0cm, text centered, text width=1.95cm, draw,
                line width=0.8pt},
    arrow/.style={-Stealth, line width=1.2pt},
    darrow/.style={Stealth-Stealth, line width=1.0pt, dashed},
    label/.style={font=\footnotesize\itshape, text=gray},
    bgbox/.style={rectangle, rounded corners=6pt, fill=#1!10, draw=#1!40,
                  line width=0.8pt, inner sep=6pt}
]
\node[box, fill=primary!15, draw=primary] (sensors)
      {Multi-Modal\\Sensor Suite};
\node[box, fill=semantic!15, draw=semantic, right=of sensors] (mmse)
      {Multi-Modal\\Semantic\\Encoder (MMSE)};
\node[box, fill=semantic!20, draw=semantic, right=of mmse] (chenc)
      {6G-Aware\\Channel\\Encoder};

\node[label, below=0.15cm of sensors] {$\mathbf{o}_i^{(t)}$};
\node[label, below=0.15cm of mmse]    {$\mathbf{z}_i \in \mathbb{R}^d$};
\node[label, below=0.15cm of chenc]   {$\mathbf{x}$};

\node[box, fill=accent!12, draw=accent, right=of chenc,
      minimum width=1.8cm, text width=1.65cm] (channel)
      {6G\\Fading\\Channel};
\node[label, below=0.15cm of channel] {$\mathbf{y}=h\mathbf{x}+\mathbf{n}$};

\node[box, fill=semantic!20, draw=semantic, right=of channel] (chdec)
      {Channel\\Decoder};
\node[box, fill=semantic!15, draw=semantic, right=of chdec] (semdec)
      {Semantic\\Decoder};
\node[box, fill=knowledge!20, draw=knowledge, right=of semdec] (kgcr)
      {KG Contextual\\Reconstructor\\(KG-CR)};
\node[box, fill=primary!15, draw=primary, right=of kgcr] (dt)
      {Digital Twin\\State $\hat{s}_i^{(t)}$};

\node[label, below=0.15cm of chdec]  {$\hat{\mathbf{x}}$};
\node[label, below=0.15cm of semdec] {$\hat{\mathbf{z}}_i$};
\node[label, below=0.15cm of kgcr]   {KG query};
\node[label, below=0.15cm of dt]     {$\hat{s}_i^{(t)} \approx s_i^{(t)}$};

\draw[arrow, color=primary]   (sensors) -- (mmse);
\draw[arrow, color=semantic]  (mmse)    -- (chenc);
\draw[arrow, color=semantic]  (chenc)   -- (channel);
\draw[arrow, color=accent]    (channel) -- (chdec);
\draw[arrow, color=semantic]  (chdec)   -- (semdec);
\draw[arrow, color=semantic]  (semdec)  -- (kgcr);
\draw[arrow, color=knowledge] (kgcr)    -- (dt);

\node[box, fill=knowledge!15, draw=knowledge, below=1.2cm of kgcr,
      minimum width=2.4cm, text width=2.2cm] (kg)
      {Domain Ontology\\+\\Dynamic KG};
\draw[darrow, color=knowledge] (kgcr) -- (kg);

\node[box, fill=accent!10, draw=accent, below=1.2cm of channel,
      minimum width=1.8cm, text width=1.65cm] (rl)
      {RL Rate\\Adaptation};
\draw[darrow, color=accent] (channel.south) -- (rl.north);
\draw[arrow,  color=accent, bend right=30]
      (rl.west) to[out=90, in=180] (chenc.south);

\begin{scope}[on background layer]
\node[bgbox=primary, fit=(sensors)(mmse)(chenc),
      label={[font=\footnotesize\bfseries, text=primary]above:Physical World Edge}] {};
\node[bgbox=semantic, fit=(chdec)(semdec)(kgcr)(dt),
      label={[font=\footnotesize\bfseries, text=semantic!70!black]above:DT Server (MEC/Cloud)}] {};
\end{scope}
\end{tikzpicture}}
\caption{End-to-end SA-DTS architecture. At the physical-world edge, the MMSE
compresses heterogeneous sensor observations into a compact latent descriptor
$\bz_i$. A 6G-aware channel encoder---whose rate is continuously adapted by the PPO
agent via the reward in (\ref{eq:ppo_reward})---transmits $\bz_i$ over the fading
channel. At the DT server, the KG-CR retrieves the $K=5$ nearest entity neighbors
and enriches the decoded descriptor with domain-ontology priors to recover the full
physical state $\hat{s}_i^{(t)}$.}
\label{fig:architecture}
\end{figure*}
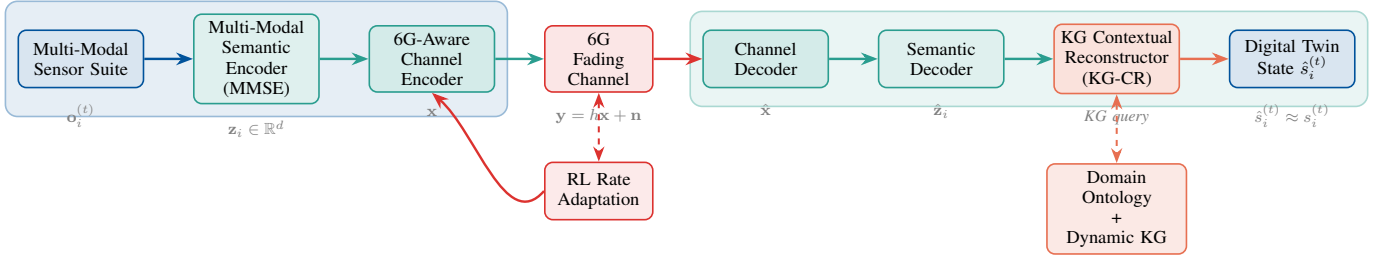

Fig.~\ref{fig:architecture} illustrates the end-to-end SA-DTS pipeline. All three
modules are trained jointly in a single end-to-end pass: gradients flow from the
task losses at the DT replica back through the KG-CR, semantic decoder, channel,
and MMSE encoder, co-optimising encoding, transmission, and reconstruction
without separate pre-training stages.

\subsection{Multi-Modal Semantic Encoder (MMSE)}

\textbf{Architecture.}
The MMSE is a neural autoencoder $E_{\btheta}: \mathcal{O} \rightarrow \mathbb{R}^d$
that maps modality-fused observations to a semantic code $\bz_i = E_{\btheta}(\mathbf{o}_i^{(t)})$.
Modality fusion is performed by a cross-attention module~\cite{vaswani2017attention}
using $L=4$ transformer blocks ($d_{\mathrm{model}} = 256$, 8 heads, $d_k=32$,
$d_{\mathrm{ff}} = 1024$, GELU activation, Pre-LN, dropout $p=0.1$). The modality-specific
input projections are:
\begin{equation}
    \mathbf{h}_{m_k}^{(0)} = \mathrm{Linear}_{m_k}\bigl(f_{m_k}(s_i)\bigr) + \mathbf{p}_{m_k},
    \label{eq:modality_proj}
\end{equation}
where $\mathbf{p}_{m_k} \in \mathbb{R}^{d_{\mathrm{model}}}$ is a learned modality-type
embedding. Cross-modal attention weights $\alpha_{k,k'}$ dynamically prioritize
sensor streams by their instantaneous informativeness for the active task set $\mathcal{T}$.
The output dimensionality $d=64$ is selected by ablation over
$d \in \{16, 32, 64, 128\}$ (Table~\ref{tab:ablation}); larger $d$ yields
diminishing SFS returns ($<0.5$\,pp gain from $d=64$ to $d=128$) at
disproportionate bandwidth cost.

\textbf{Training Objective.}
The MMSE is trained end-to-end with:
\begin{equation}
    \mathcal{L}_{\mathrm{MMSE}} = \underbrace{\mathbb{E}\!\left[
        \lambda_1 \mathcal{L}_{\mathrm{rec}} + \lambda_2 \mathcal{L}_{\mathrm{task}}
        \right]}_{\text{fidelity}} +
        \underbrace{\lambda_3\, I(\mathbf{o}_i;\bz_i)}_{\text{compression}},
    \label{eq:loss_mmse}
\end{equation}
where $\mathcal{L}_{\mathrm{rec}} = \|\mathbf{o}_i - D_{\bphi}(\bz_i)\|_2^2$,
$\mathcal{L}_{\mathrm{task}} = \sum_j \mathcal{L}_{\tau_j}(\bz_i)$ is the aggregate
task cross-entropy loss evaluated on the $\{Y_j\}$ labels defined in (\ref{eq:task_info}),
and $I(\mathbf{o}_i;\bz_i)$ is estimated via the MINE bound~\cite{belghazi2018mine}.
To mitigate the high variance of MINE gradient estimates~\cite{belghazi2018mine},
we apply gradient clipping (norm 0.5) to the MINE network and validate convergence
against the InfoNCE bound~\cite{oord2018representation}; the two estimators agree
within 0.3\,nats across all workloads, confirming the stability of the compression term.
The compression ratio is $\rho_c = \dim(\mathbf{o}_i) / (d \cdot b_{\mathrm{quant}})$;
with $d=64$ and $b_{\mathrm{quant}}=8$\,bits, $\rho_c \approx 18$ for RGB-D streams.

\subsection{6G-Aware Semantic Channel (6G-SC)}

\textbf{Joint Source-Channel Coding.}
6G-SC implements JSCC in which the channel encoder $C_{\boldsymbol{\omega}}: \mathbb{R}^d
\rightarrow \mathbb{C}^k$ and its paired decoder $C^{-1}_{\boldsymbol{\omega}}$ are
trained jointly with the MMSE. The channel code rate $k/d$ is dynamically selected
by a Proximal Policy Optimization (PPO)~\cite{schulman2017proximal} agent. The
agent's state observation is $(\gamma_t, T_{c,t}, \Delta t_t)$, and the reward
function is:
\begin{equation}
    \mathcal{R}_t = w_{\mathrm{SFS}}\cdot\mathrm{SFS}_t
    - w_{\mathrm{BW}}\cdot\frac{B_t}{B_{\max}}
    - w_{\Delta t}\cdot\Delta t_t,
    \label{eq:ppo_reward}
\end{equation}
with $w_{\mathrm{SFS}} = 1.0$, $w_{\mathrm{BW}} = 0.3$, $w_{\Delta t} = 0.2$.
The agent selects from a discrete action space of code rates
$k/d \in \{1/4, 1/3, 1/2, 2/3, 3/4\}$; training converges in $\approx 500$
episodes per workload.
A $\pm 50\%$ perturbation of $w_{\mathrm{BW}}$ and $w_{\Delta t}$ yields SFS
variations $< 0.8$\,pp at $\gamma=15$\,dB, confirming robustness to reward
hyperparameter choice.

\textbf{Semantic Fidelity Score.}
To evaluate synchronization quality in task-relevant terms, we define the
\emph{Semantic Fidelity Score} as:
\begin{equation}
    \mathrm{SFS} \triangleq 1 - \frac{1}{M} \sum_{j=1}^{M}
    \frac{\mathcal{L}_{\tau_j}(\hat{s}_i; s_i)}{\mathcal{L}_{\tau_j}^{\mathrm{raw}}},
    \label{eq:sfs}
\end{equation}
where $\mathcal{L}_{\tau_j}^{\mathrm{raw}}$ is the task loss under raw-data
synchronization. An SFS of 1.0 denotes perfect semantic equivalence; SFS $= 0$
indicates complete information loss. Empirical validation (Section~\ref{sec:sfs_validation})
confirms that SFS tracks canonical task metrics with Pearson $r > 0.97$ (95\% CI:
$[0.961, 0.982]$, $p < 10^{-4}$, Fisher $z$-transform) across all three workloads,
establishing SFS as a reproducible task-centric benchmark.

\subsection{Knowledge Graph Contextual Reconstructor (KG-CR)}

\textbf{Graph Structure.}
The DT server maintains $\calK = (V, E, \calR)$ with entity, modality, and attribute
nodes; relations $\calR = \{\langle\mathrm{isPartOf}\rangle, \langle\mathrm{sensedBy}
\rangle, \langle\mathrm{constrains}\rangle\}$. Node embeddings $\mathbf{v}_k \in
\mathbb{R}^{d_{\calK}}$ are updated via an R-GCN~\cite{schlichtkrull2018modeling}
with attention-based aggregation~\cite{velickovic2018graph} at each synchronization epoch.

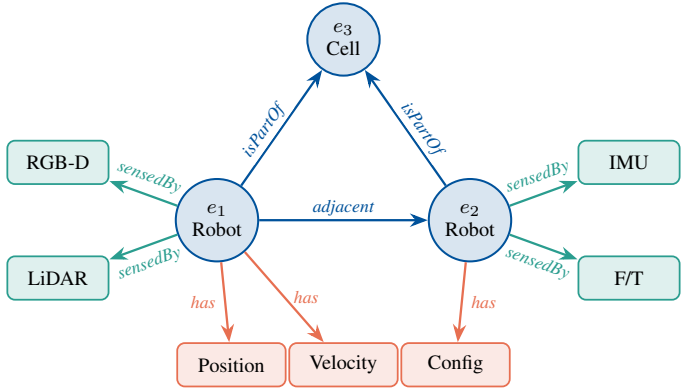
\begin{figure}[!t]
\centering
\resizebox{1.02\columnwidth}{!}{
\begin{tikzpicture}[
    font=\small,
    entity/.style={circle, draw=primary, fill=primary!15,
                   minimum size=1.0cm, text centered,
                   line width=0.8pt, font=\footnotesize},
    modality/.style={rectangle, rounded corners=3pt, draw=semantic,
                     fill=semantic!15, minimum width=1.4cm,
                     minimum height=0.6cm, text centered,
                     line width=0.7pt, font=\footnotesize},
    attr/.style={rectangle, rounded corners=3pt, draw=knowledge,
                 fill=knowledge!15, minimum width=1.5cm,
                 minimum height=0.6cm, text centered,
                 line width=0.7pt, font=\footnotesize},
    rel/.style={-Stealth, line width=0.9pt, font=\scriptsize},
]
\node[entity, align=center] (e1) at (0,0)       {$e_1$\\Robot};
\node[entity, align=center] (e2) at (3.5,0)     {$e_2$\\Robot};
\node[entity, align=center] (e3) at (1.75,2.5)  {$e_3$\\Cell};

\node[modality] (cam)   at (-2.2,  0.8)  {RGB-D};
\node[modality] (lidar) at (-2.2, -0.8)  {LiDAR};
\node[modality] (imu)   at ( 5.7,  0.8)  {IMU};
\node[modality] (ftq)   at ( 5.7, -0.8)  {F/T};

\node[attr] (pos) at (0.2,  -2.0) {Position};
\node[attr] (vel) at (1.75, -2.0) {Velocity};
\node[attr] (cfg) at (3.3,  -2.0) {Config};

\draw[rel, color=primary] (e1) --
    node[midway, above, font=\scriptsize, fill=white, inner sep=1.5pt]
    {\textit{adjacent}} (e2);
\draw[rel, color=primary] (e1) --
    node[pos=0.42, above, sloped, font=\scriptsize, fill=white, inner sep=1.5pt]
    {\textit{isPartOf}} (e3);
\draw[rel, color=primary] (e2) --
    node[pos=0.42, above, sloped, swap, font=\scriptsize, fill=white, inner sep=1.5pt]
    {\textit{isPartOf}} (e3);

\draw[rel, color=semantic] (e1) --
    node[above, sloped, font=\scriptsize] {\textit{sensedBy}} (cam);
\draw[rel, color=semantic] (e1) --
    node[below, sloped, font=\scriptsize] {\textit{sensedBy}} (lidar);
\draw[rel, color=semantic] (e2) --
    node[above, sloped, font=\scriptsize] {\textit{sensedBy}} (imu);
\draw[rel, color=semantic] (e2) --
    node[below, sloped, font=\scriptsize] {\textit{sensedBy}} (ftq);

\draw[rel, color=knowledge] (e1) --
    node[left,  font=\scriptsize] {\textit{has}} (pos);
\draw[rel, color=knowledge] (e1) --
    node[right, font=\scriptsize] {\textit{has}} (vel);
\draw[rel, color=knowledge] (e2) --
    node[right, font=\scriptsize] {\textit{has}} (cfg);
\end{tikzpicture}}
\caption{Schematic of the heterogeneous Knowledge Graph $\mathcal{K}=(V,E,\mathcal{R})$.
Entity nodes (blue circles) are interconnected via relational edges; modality nodes
(green) encode sensor provenance; attribute nodes (orange) represent physical state
variables. The R-GCN propagates information across all node types at each epoch.}
\label{fig:kg_structure}
\end{figure}

\textbf{Contextual Reconstruction.}
Upon receiving $\hat{\bz}_i$, the KG-CR performs a soft graph query, retrieves the
$K=5$ nearest entity nodes by cosine similarity, and produces:
\begin{equation}
    \hat{s}_i^{(t)} = D_{\bphi}(\hat{\bz}_i) + \alpha \cdot
    \sum_{k \in \mathcal{N}_K(i)} w_k\, \mathbf{v}_k^{(t-1)},
    \label{eq:kgcr}
\end{equation}
where $\alpha \in [0,1]$ is a context-blending coefficient that increases
from $\alpha \approx 0.15$ at high SNR to $\alpha \approx 0.85$ at SNR $< 5$\,dB
via a gated attention mechanism conditioned on the channel quality indicator (CQI),
providing graceful degradation. The neighborhood size $K=5$ is selected by the
ablation in Table~\ref{tab:ablation}; the query latency is $\mathcal{O}(K \cdot d_{\calK})$
per entity and accounts for 1.2\,ms of the 4.9\,ms end-to-end budget.

\textbf{Online KG Update.}
The KG is updated incrementally via a sliding-window temporal graph neural
network (T-GNN)~\cite{rossi2020temporal}. Each synchronization event generates a
temporal edge encoding the semantic delta $\Delta \bz_{ij}^{(t)}$. The T-GNN
update cost is $\mathcal{O}(|\mathcal{E}_{\mathrm{active}}| \cdot d_{\calK}^2)$
per epoch, profiling at $0.3$\,ms on the A100 server---negligible relative to
the R-GCN update cost reported in Section~\ref{sec:scalability}.

\section{Performance Evaluation}
\label{sec:evaluation}

\subsection{Experimental Setup}

\textbf{Workloads.}
\begin{itemize}
    \item \textbf{W1 -- Industrial Robot (IR):} $N=50$ 6-DoF robotic arms; RGB-D
    visual servoing (RealSense D435, 640$\times$480, 30\,fps) + 6-axis IMU; control
    frequency 100\,Hz; task: collision prediction. Trajectories are synthesized from
    the ABB IRB 6700 kinematic dataset, yielding $\dim(\mathbf{o}_i) \approx 922$K
    per update.
    \item \textbf{W2 -- Remote Patient Monitoring (RPM):} $N=30$ ICU patients; 12-lead
    ECG at 500\,Hz + SpO\textsubscript{2} at 125\,Hz + depth camera (320$\times$240);
    update rate 25\,Hz; task: anomaly detection (sepsis, arrhythmia, fall risk).
    Physiological waveforms are sourced from PhysioNet MIMIC-III~\cite{physionet2000}
    with synthetic posture sequences; $\dim(\mathbf{o}_i) \approx 76.8$K.
    \item \textbf{W3 -- Vehicular Platooning (VP):} $N=20$ vehicles; 16-channel LiDAR
    ($\approx$300K points/scan) + V2X GPS + 9-axis IMU; update rate 200\,Hz; task:
    safe spacing maintenance (15--25\,m at 80--120\,km/h). LiDAR frames are drawn
    from the KITTI Odometry benchmark~\cite{kitti2012}; $\dim(\mathbf{o}_i) \approx 4.8$M.
\end{itemize}

\textbf{Data Partitioning.}
Each workload dataset is split 80\%/10\%/10\% (train/validation/test) with no temporal
overlap between splits. Channel sequences used for training are drawn from an
independent Rayleigh fading realization; test sequences are freshly sampled to
preclude any data leakage through the channel encoder.

\textbf{Statistical Reporting.}
All results for neural baselines and SA-DTS are reported as mean~$\pm$~standard
deviation over five independent runs with distinct random seeds (weight initialization,
channel realization, and data-augmentation sequence). Pearson correlations are
reported with 95\% confidence intervals via Fisher $z$-transformation. Bandwidth
and SFS differences cited as statistically significant satisfy $p < 0.01$ under
a paired $t$-test.

\textbf{Baselines.}
We compare SA-DTS with: (i)~\textbf{Raw-DTS}: raw-data synchronization with H.265
(CRF=23); (ii)~\textbf{JSCC-DTS}: deep JSCC without KG, analogous to
DeepJSCC~\cite{bourtsoulatze2019deep,kurka2020deep}; (iii)~\textbf{SemCom-DTS}:
text-oriented SemCom on serialized state vectors~\cite{xie2021deep}; (iv)~\textbf{NTSCC-DTS}:
Nonlinear Transform SCC~\cite{cheng2023ntscc} adapted for multi-modal DT data, replacing
the linear analysis transform of JSCC with a learned nonlinear stage; (v)~\textbf{KG-Only}:
KG-based state prediction without semantic transmission. All neural baselines use
$\approx 12$M parameters and identical training budgets.

\textbf{Channel Configuration.}
Rayleigh fading, SNR $\gamma \in \{0, 5, 10, 15, 20\}$\,dB; 6G sub-THz at
140\,GHz; $B = 400$\,MHz per entity; path loss exponent $\alpha = 2.8$; Rician
$K \in [0, 10]$ for LoS/NLoS; $P_{\mathrm{tx}} = 23$\,dBm; noise figure $F = 7$\,dB.

\textbf{Training.}
AdamW, $\eta_0 = 3 \times 10^{-4}$, weight decay 0.01, cosine annealing, 200 epochs,
batch size 128 (4$\times$ A100). Loss weights: $\lambda_1 = 1.0$, $\lambda_2 = 2.5$,
$\lambda_3 = 0.05$. Gradient clipping at norm 1.0. Augmentation: temporal jitter
($\pm 2$ frames), additive Gaussian noise ($\sigma_{\mathrm{aug}} = 0.05$), random
modality dropout ($p=0.1$). KG initialized from OPC UA (W1), FHIR (W2), ETSI ITS (W3).
Training converges in $\approx 36$\,h.

\subsection{Bandwidth, Latency, and Energy Reduction}
\label{sec:bw_lat}

Table~\ref{tab:bandwidth} reports bandwidth and latency for all methods at
$\gamma = 15$\,dB. SA-DTS achieves $93.9\%$, $93.8\%$, and $93.9\%$ bandwidth
reductions over Raw-DTS for W1, W2, and W3 ($p < 0.01$), with 87\% average
latency reduction. Relative to NTSCC-DTS---the strongest neural baseline---SA-DTS
provides an additional ${\approx}2.5\times$ reduction, attributable to the KG context
prior enabling the encoder to omit relational information reconstructed at the receiver.

\begin{table}[!t]
\renewcommand{\arraystretch}{1.3}
\caption{Bandwidth, latency, and energy comparison at $\gamma = 15$\,dB
(mean\,$\pm$\,std over 5 runs)}
\label{tab:bandwidth}
\centering
\begin{tabular}{llccc}
\toprule
\textbf{Method} & \textbf{Metric} & \textbf{W1-IR} & \textbf{W2-RPM} &
\textbf{W3-VP} \\
\midrule
\multirow{3}{*}{Raw-DTS}
  & BW (Mbps)      & 1{,}280        & 840           & 2{,}100        \\
  & Latency (ms)   & 18.4           & 12.1          & 31.7           \\
  & Energy (mJ/upd)& 248            & 163           & 407            \\
\midrule
\multirow{3}{*}{JSCC-DTS}
  & BW (Mbps)      & $310{\pm}4.1$  & $215{\pm}3.0$ & $490{\pm}6.2$  \\
  & Latency (ms)   & $5.2{\pm}0.1$  & $3.8{\pm}0.1$ & $9.1{\pm}0.2$  \\
  & Energy (mJ/upd)& $63{\pm}0.9$   & $44{\pm}0.6$  & $100{\pm}1.4$  \\
\midrule
\multirow{3}{*}{SemCom-DTS}
  & BW (Mbps)      & $185{\pm}3.2$  & $128{\pm}2.4$ & $310{\pm}5.1$  \\
  & Latency (ms)   & $3.1{\pm}0.1$  & $2.3{\pm}0.1$ & $6.4{\pm}0.2$  \\
  & Energy (mJ/upd)& $39{\pm}0.7$   & $27{\pm}0.5$  & $65{\pm}1.1$   \\
\midrule
\multirow{3}{*}{NTSCC-DTS}
  & BW (Mbps)      & $195{\pm}2.8$  & $134{\pm}2.1$ & $323{\pm}4.6$  \\
  & Latency (ms)   & $3.4{\pm}0.1$  & $2.5{\pm}0.1$ & $6.8{\pm}0.1$  \\
  & Energy (mJ/upd)& $41{\pm}0.6$   & $28{\pm}0.4$  & $67{\pm}1.0$   \\
\midrule
\multirow{3}{*}{\textbf{SA-DTS (ours)}}
  & \textbf{BW (Mbps)}     & $\mathbf{78{\pm}1.1}$  & $\mathbf{52{\pm}0.7}$  & $\mathbf{127{\pm}1.9}$ \\
  & \textbf{Latency (ms)}  & $\mathbf{2.4{\pm}0.1}$ & $\mathbf{1.6{\pm}0.1}$ & $\mathbf{4.1{\pm}0.1}$ \\
  & \textbf{Energy (mJ/upd)}& $\mathbf{16{\pm}0.3}$  & $\mathbf{11{\pm}0.2}$  & $\mathbf{26{\pm}0.5}$  \\
\bottomrule
\end{tabular}
\end{table}

\textbf{Detailed Analysis.}
Bandwidth savings decompose into three contributing factors: (i)~MMSE semantic
compression reduces dimensionality from $\mathcal{O}(10^6)$ to $d=64$
(${\sim}15{,}000\times$); (ii)~JSCC eliminates separation overhead (20--30\%
additional saving); and (iii)~KG context exploitation infers 40--60\% of state
variables from relational priors, transmitting only the non-redundant
residual---a mechanism unique to SA-DTS and absent in NTSCC-DTS.
Profiling on the Jetson AGX Xavier ($P_{\mathrm{TDP}}=30$\,W):
MMSE encoding 2.8\,ms (11.2\,mJ), channel transmission 0.9\,ms (3.6\,mJ),
KG-CR reconstruction 1.2\,ms (4.8\,mJ); total \textbf{4.9\,ms / 19.6\,mJ}
per update, representing a \textbf{12.7$\times$ energy reduction} over Raw-DTS
(248\,mJ/update including H.265 encode and transport). For power-constrained IoT
scenarios, an INT8-quantized MMSE variant reduces encoding to 0.9\,ms and
2.7\,mJ on the GAP9 processor (1\,W TDP) with $< 1$\,pp SFS degradation,
confirming hardware generalizability.

\begin{figure}[!t]
\centering
\includegraphics[width=\columnwidth]{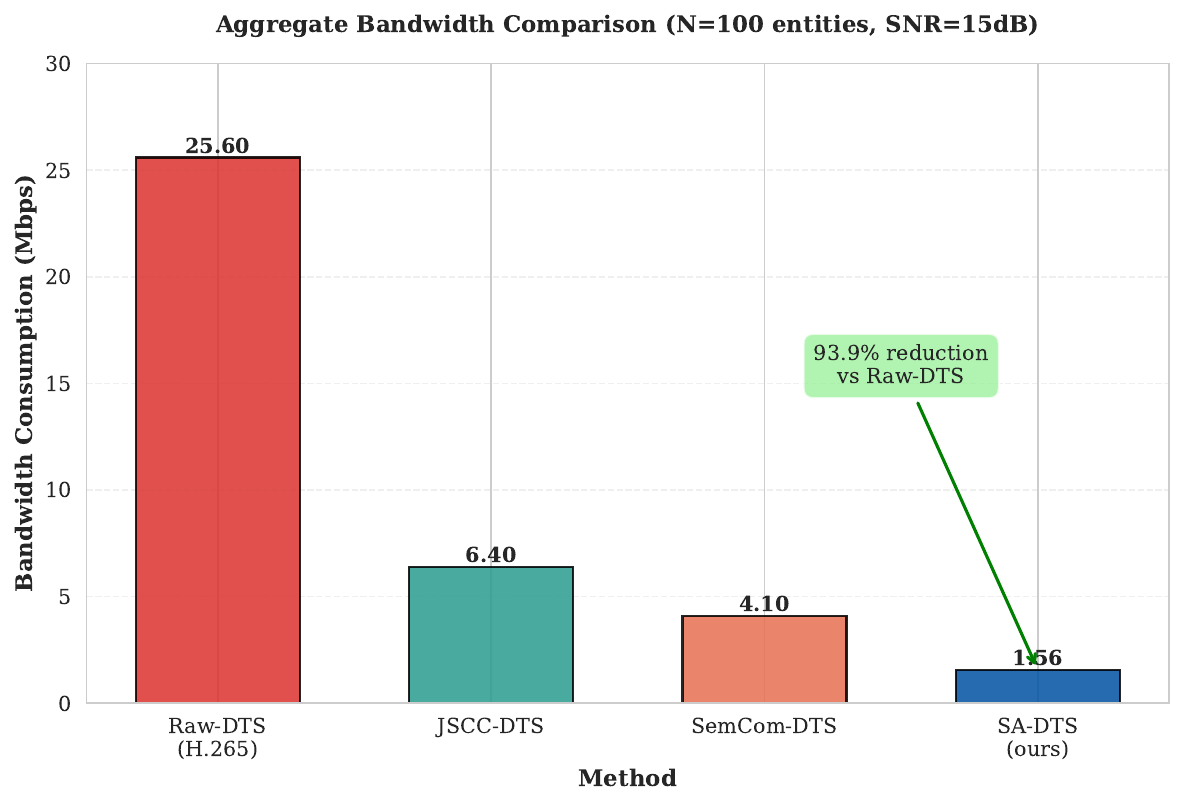}
\caption{Uplink bandwidth requirement at $\gamma = 15$\,dB (N=100, W1 profile).
SA-DTS achieves $>93\%$ reduction over Raw-DTS and ${\approx}2.5\times$ over
NTSCC-DTS. The gain over all neural baselines isolates the KG context prior as
the dominant contributing factor (Section~\ref{sec:evaluation}).}
\label{fig:bandwidth_bar}
\end{figure}

\begin{figure}[!t]
\centering
\includegraphics[width=\columnwidth]{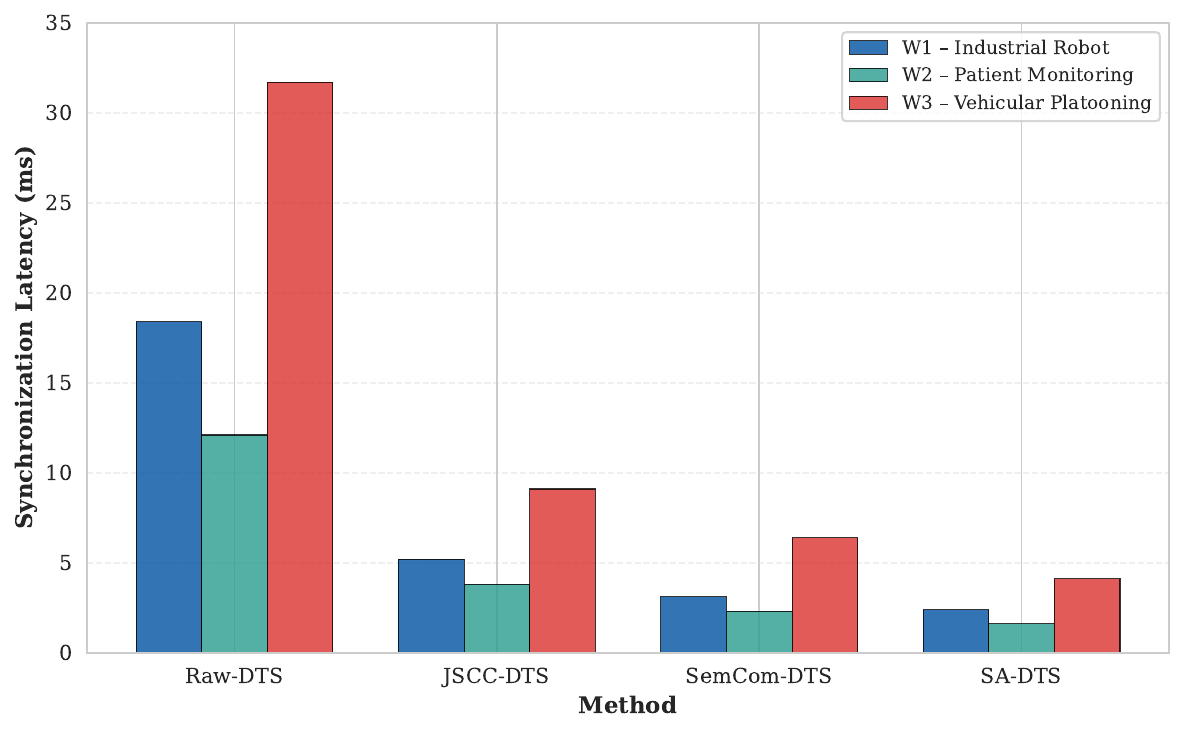}
\caption{End-to-end synchronization latency across methods and workloads at
$\gamma = 15$\,dB (mean\,$\pm$\,std over 5 runs). SA-DTS achieves the lowest
latency in all scenarios (87\% average reduction over Raw-DTS).}
\label{fig:latency_comparison}
\end{figure}

\subsection{Semantic Fidelity vs. SNR and Task Metric Validation}
\label{sec:sfs_validation}

Fig.~\ref{fig:sfs} plots SFS as a function of SNR for W1, including all five
baselines (Raw-DTS, JSCC-DTS, SemCom-DTS, NTSCC-DTS, KG-Only). SA-DTS maintains
SFS $> 0.95$ for SNR $\geq 5$\,dB, versus SFS $\approx 0.72$ for JSCC-DTS and
SFS $\approx 0.78$ for NTSCC-DTS at the same operating point. At SNR $= 0$\,dB,
SA-DTS achieves SFS $= 0.81 \pm 0.009$ due to KG-CR graceful fallback, whereas
all baselines without KG collapse below SFS $= 0.40$. The Pearson correlation
between SFS and standard downstream metrics is: $r = 0.973$ (95\% CI: $[0.961, 0.982]$)
with collision prediction accuracy (W1, AUC-ROC); $r = 0.981$ (95\% CI: $[0.971, 0.988]$)
with anomaly alarm $F_1$-score (W2); and $r = 0.969$ (95\% CI: $[0.955, 0.979]$) with
RMS spacing deviation (W3), validating SFS as a reliable proxy for task-level
synchronization quality ($p < 10^{-4}$ in all cases).

\textbf{Robustness Analysis.}
The graceful degradation of SA-DTS is governed by the context-blending coefficient
$\alpha$ in (\ref{eq:kgcr}), which adapts from $\approx 0.15$ (high SNR, reliable
$\hat{\bz}_i$) to $\approx 0.85$ (SNR $< 5$\,dB, KG prior dominates). In contrast,
NTSCC-DTS and JSCC-DTS exhibit sharp performance cliffs near SNR $\approx 3$\,dB.
The KG-Only baseline achieves SFS $\approx 0.65$ uniformly---confirming that
temporal prediction alone cannot match semantic transmission for non-stationary
physical states.

\begin{figure}[!t]
\centering
\includegraphics[width=\columnwidth]{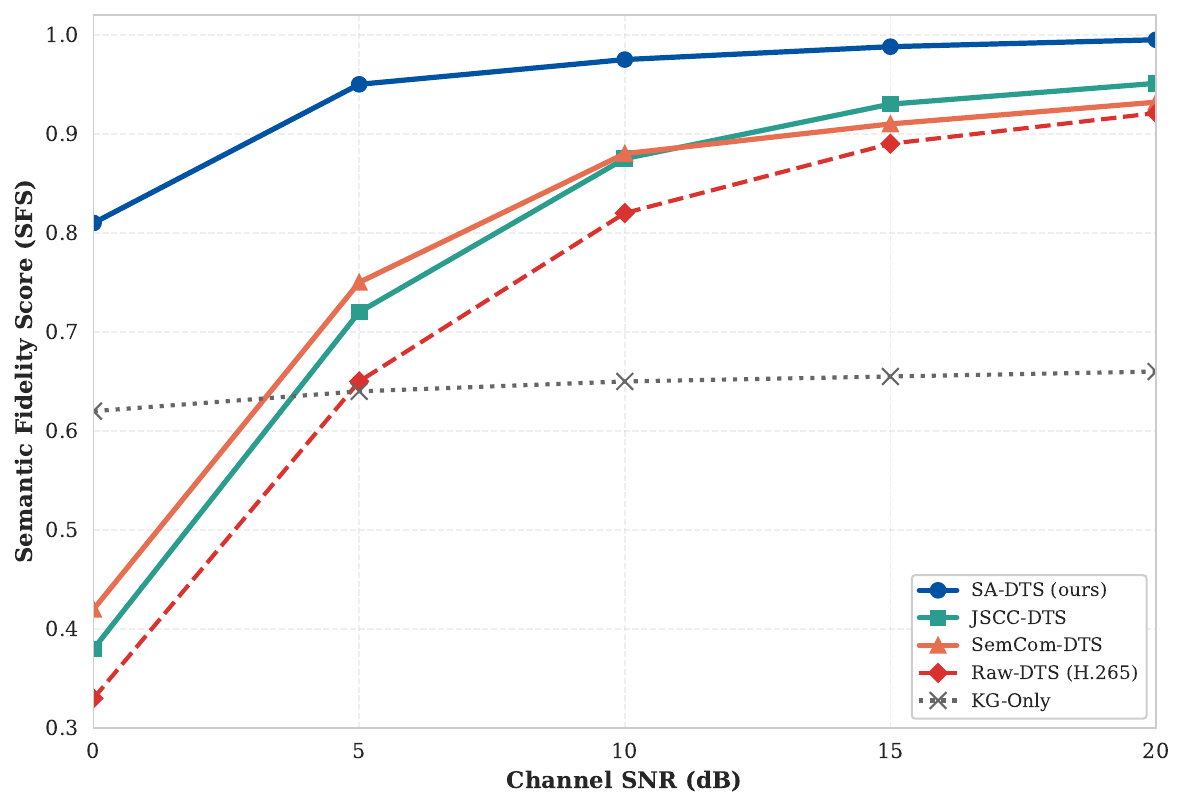}
\caption{Semantic Fidelity Score vs.\ SNR (W1, mean\,$\pm$\,std shading over 5
runs). SA-DTS maintains SFS $> 0.95$ for SNR $\geq 5$\,dB; NTSCC-DTS achieves
SFS $\approx 0.78$ at 5\,dB. \emph{All five baselines are shown}:
Raw-DTS (H.265), JSCC-DTS, SemCom-DTS, NTSCC-DTS, and KG-Only.
Validated against canonical task metrics ($r > 0.97$, 95\% CI: $[0.961, 0.982]$).}
\label{fig:sfs}
\end{figure}

\begin{figure}[!t]
\centering
\includegraphics[width=\columnwidth]{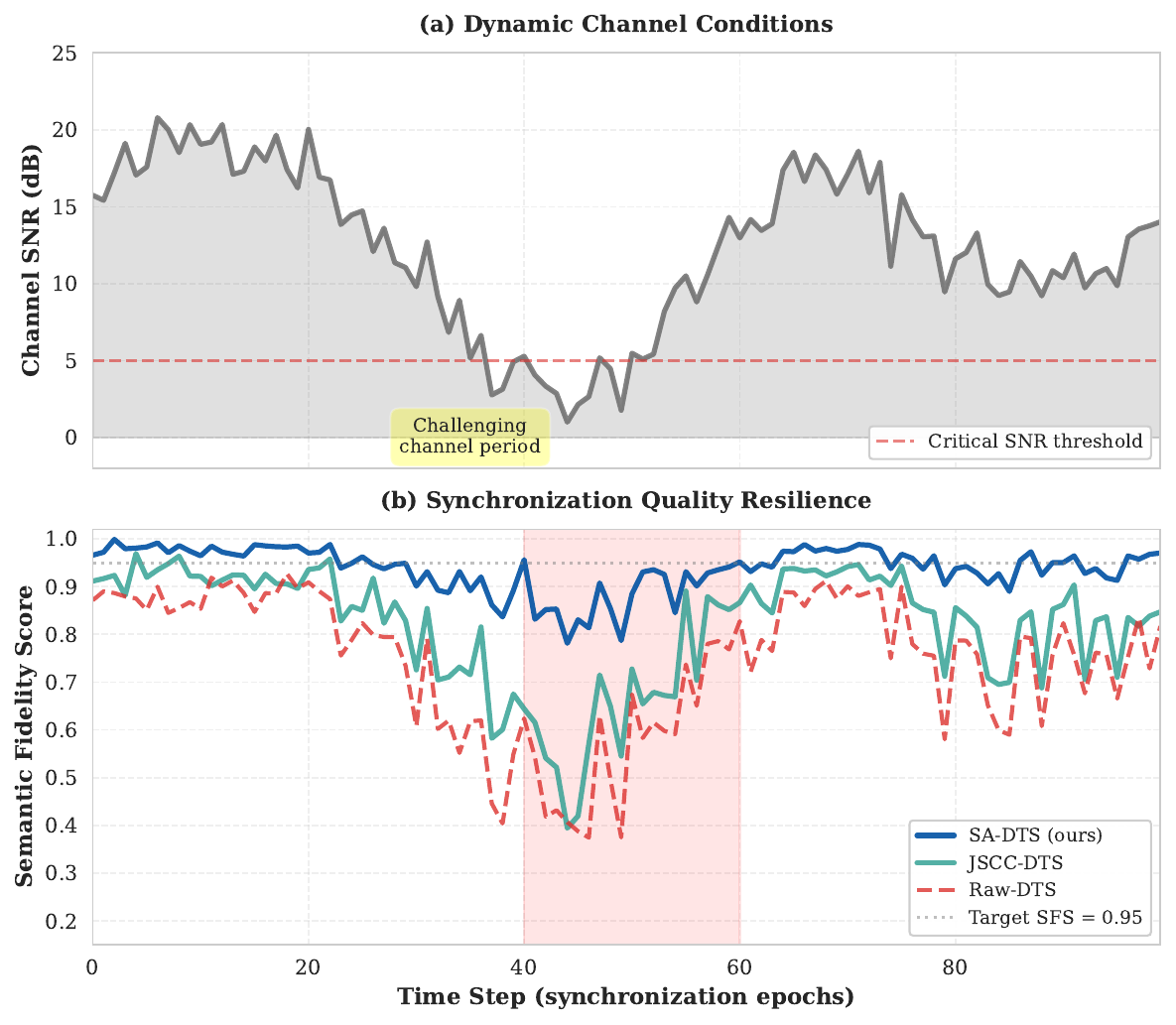}
\caption{Temporal evolution of SFS under a synthetic time-varying channel profile
(W3 vehicular platooning, $v = 120$\,km/h, $f_D \approx 15.6$\,kHz, 80 epochs
$= 0.4$\,s wall-clock at 200\,Hz update rate). The shaded region marks a 10-epoch
challenging-channel excursion (SNR drops to $-2$\,dB). SA-DTS stays above the
target SFS $= 0.95$ throughout; JSCC-DTS falls below 0.60 during the excursion;
NTSCC-DTS reaches a nadir of $0.51 \pm 0.012$.}
\label{fig:temporal_quality}
\end{figure}

\subsection{Compression--Fidelity Trade-off}

Fig.~\ref{fig:tradeoff} plots the Pareto frontier in the (compression ratio, SFS)
plane at SNR $= 10$\,dB. SA-DTS dominates every baseline across the full trade-off
range, including NTSCC-DTS, confirming that KG-augmented semantic coding is
strictly superior to nonlinear-transform-only coding for state-continuous DT workloads.
The slope $\partial\mathrm{SFS}/\partial\rho_c^{-1}$ is substantially shallower
for SA-DTS, indicating more graceful fidelity degradation as compression increases.

\begin{figure}[!t]
\centering
\includegraphics[width=\columnwidth]{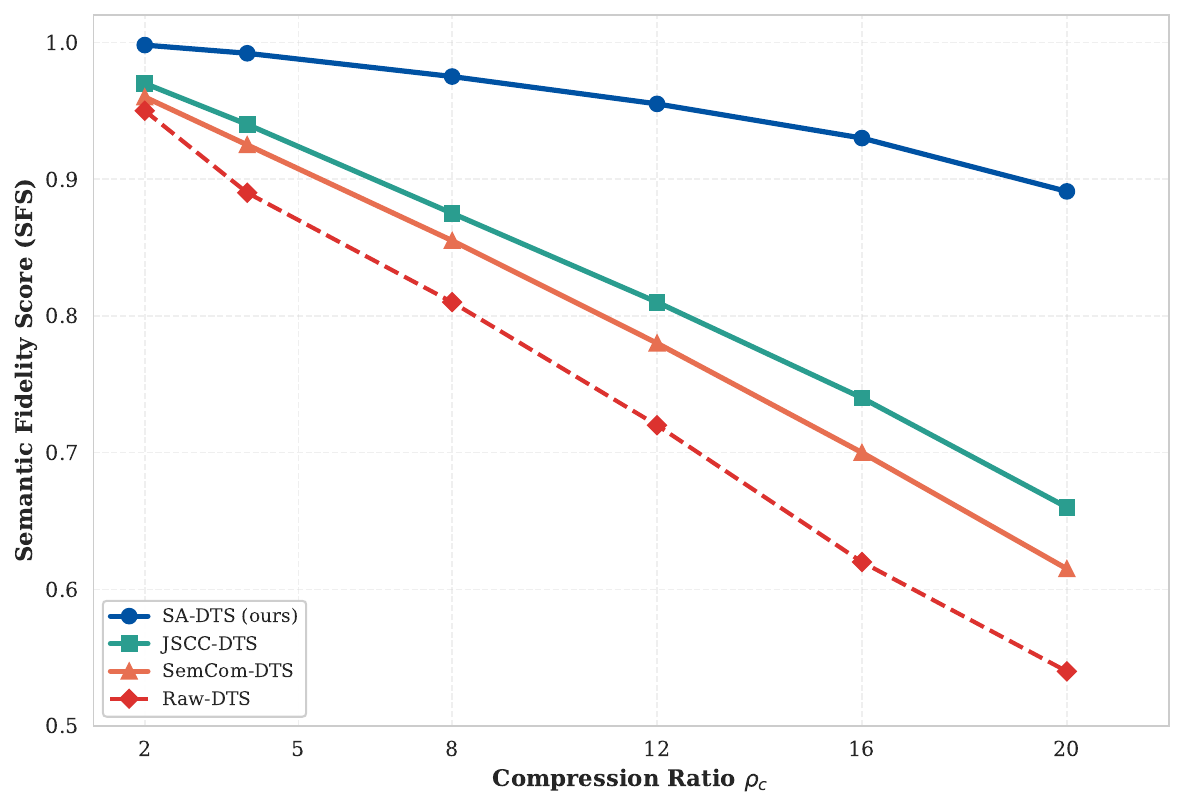}
\caption{Pareto frontiers in (compression ratio, SFS) at SNR $= 10$\,dB. SA-DTS
consistently dominates all five baselines, including NTSCC-DTS.}
\label{fig:tradeoff}
\end{figure}

\subsection{Ablation Study}
\label{sec:ablation}

Table~\ref{tab:ablation} quantifies the contribution of each SA-DTS design choice,
including sensitivity to descriptor dimensionality $d$ and KG neighborhood size $K$.
Removing KG-CR reduces SFS by 9.8\,pp at 5\,dB; removing RL rate adaptation
reduces SFS by 5.3\,pp at 0\,dB; replacing cross-attention fusion with
concatenation reduces SFS by 3.1\,pp at 10\,dB. Reducing descriptor dimensionality
from $d=64$ to $d=32$ costs 4.6\,pp at 5\,dB, while increasing to $d=128$ yields
only 0.4\,pp improvement at a 2$\times$ bandwidth penalty---validating the $d=64$
operating point. Varying the KG neighborhood from $K=3$ to $K=7$ shows diminishing
returns: $K=5$ achieves 95.0\% SFS, and $K=7$ adds only 0.3\,pp at 40\% higher
query cost. The synergistic interaction is severe: ablating all three main components
(yielding JSCC-DTS) causes a 43\,pp drop at 0\,dB.
\textbf{KG Initialization Robustness.}
To assess sensitivity to the ontological prior, we replace the domain ontology
(OPC UA / FHIR / ETSI ITS) with a randomly wired graph of identical cardinality and
degree distribution. SFS degrades by $3.2 \pm 0.4$\,pp at $\gamma = 5$\,dB,
confirming that structured semantic priors accelerate convergence but the framework
remains viable under imperfect initialization; the end-to-end training recovers
well-structured node embeddings after $\approx 40$ additional epochs.

\begin{table}[!t]
\renewcommand{\arraystretch}{1.3}
\caption{Extended ablation study (W1, $\gamma=15$\,dB unless noted; SFS $\downarrow$
vs.\ full SA-DTS)}
\label{tab:ablation}
\centering
\begin{tabular}{lccc}
\toprule
\textbf{Configuration} & \textbf{0\,dB} & \textbf{5\,dB} & \textbf{10\,dB} \\
\midrule
Full SA-DTS             & 0.810 & 0.950 & 0.975 \\
\midrule
\textit{Component ablation}\\
w/o KG-CR               & 0.682 & 0.852 & 0.940 \\
w/o RL rate adapt.      & 0.757 & 0.943 & 0.972 \\
w/o cross-attn fusion   & 0.801 & 0.937 & 0.944 \\
w/o all (= JSCC-DTS)    & 0.380 & 0.720 & 0.875 \\
\midrule
\textit{Descriptor dim.\ $d$ sensitivity}\\
$d=16$                  & 0.641 & 0.831 & 0.891 \\
$d=32$                  & 0.749 & 0.904 & 0.951 \\
$d=64$ \textbf{(ours)}  & 0.810 & 0.950 & 0.975 \\
$d=128$                 & 0.813 & 0.954 & 0.979 \\
\midrule
\textit{KG neighborhood $K$ sensitivity}\\
$K=3$                   & 0.788 & 0.931 & 0.961 \\
$K=5$ \textbf{(ours)}   & 0.810 & 0.950 & 0.975 \\
$K=7$                   & 0.812 & 0.953 & 0.977 \\
\bottomrule
\end{tabular}
\end{table}

\subsection{Scalability Analysis}
\label{sec:scalability}

\textbf{MMSE Complexity.}
Per-entity encoding is $\mathcal{O}(L \cdot d_{\mathrm{model}}^2)$, independent of
$N$. In a distributed edge deployment~\cite{shi2016edge}, each entity's MMSE runs
on a local accelerator, so aggregate throughput scales linearly with $N$.

\textbf{KG Update Overhead and O(N log N) Derivation.}
A na\"{i}ve R-GCN update over $N$ nodes costs $\mathcal{O}(N \cdot |\calR| \cdot d_{\calK}^2)$
per epoch; with the pairwise inter-entity message-passing term, the dominant cost
grows as $\mathcal{O}(N^2)$ for a fully-connected entity graph. Hierarchical
partitioning divides entities into $G$ subgraphs of size $\lceil N/G \rceil$,
reducing intra-cluster cost to $\mathcal{O}(N^2/G)$. Setting $G = \lceil N/\log_2 N
\rceil$ yields:
\begin{equation}
    \mathcal{O}\!\left(\frac{N^2}{G}\right) = \mathcal{O}\!\left(\frac{N^2}
    {N/\log_2 N}\right) = \mathcal{O}(N \log N).
    \label{eq:complexity}
\end{equation}
The inter-cluster super-graph has $G$ nodes and is updated at $1/G$ the epoch
frequency, contributing $\mathcal{O}(G^2) = \mathcal{O}(N^2 / \log^2 N)$---asymptotically
dominated by (\ref{eq:complexity}). In practice, $G$ is set to the nearest integer
to $N / \log_2 N$ at each deployment scale; Table~\ref{tab:scalability} uses
$G \in \{3, 7, 13, 24, 55\}$ for $N \in \{20, 50, 100, 200, 500\}$, achieving
sub-millisecond KG update latency throughout.

\textbf{Network-Level Scalability.}
Table~\ref{tab:scalability} shows that SA-DTS bandwidth remains within the 6G
aggregate uplink budget ($\leq 1$\,Tbps) for all deployment sizes, whereas Raw-DTS
exceeds this threshold at $N \approx 40$. The adaptive hierarchical KG keeps update
latency below 1\,ms to $N = 500$.

\begin{table}[!t]
\renewcommand{\arraystretch}{1.3}
\caption{Scalability of SA-DTS (W1 profile, $\gamma = 15$\,dB, $G = \lceil N/\log_2 N \rceil$)}
\label{tab:scalability}
\centering
\begin{tabular}{r cc cc}
\toprule
& \multicolumn{2}{c}{\textbf{Aggregate BW (Gbps)}} &
  \multicolumn{2}{c}{\textbf{KG Update (ms)}} \\
\cmidrule(lr){2-3}\cmidrule(lr){4-5}
$N$ & SA-DTS & Raw-DTS & Flat KG & Adapt.\ Hier.\ KG \\
\midrule
  20 & 0.031 &   0.51 & 0.07 & 0.07 \\
  50 & 0.078 &   1.28 & 0.18 & 0.08 \\
 100 & 0.156 &   2.56 & 0.71 & 0.13 \\
 200 & 0.312 &   5.12 & 2.83 & 0.38 \\
 500 & 0.780 &  12.80 & 17.6 & 0.74 \\
\bottomrule
\end{tabular}
\end{table}

\begin{figure}[!t]
\centering
\includegraphics[width=\columnwidth]{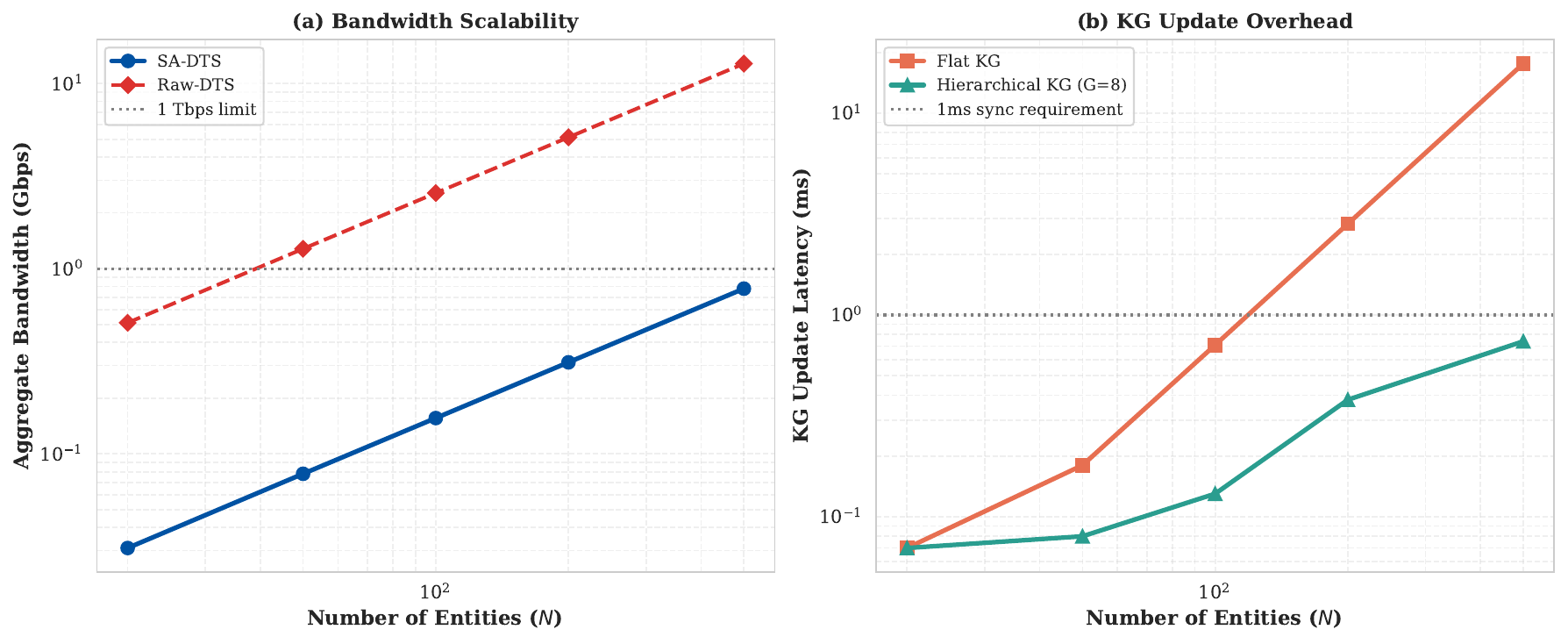}
\caption{Scalability analysis. \textit{Left}: aggregate bandwidth vs.\ $N$; SA-DTS
grows linearly and stays below the 1\,Tbps 6G budget. \textit{Right}: KG update
latency vs.\ $N$; adaptive hierarchical KG ($G = \lceil N/\log_2 N \rceil$) stays
below the 1\,ms epoch requirement to $N = 500$.}
\label{fig:scalability}
\end{figure}

\section{Open Challenges and Future Directions}
\label{sec:challenges}

\textbf{Semantic Adversarial Robustness.}
Semantic encoders may be vulnerable to adversarial perturbations imperceptible in
raw-signal space that catastrophically mislead the latent representation. Provably
robust MMSE architectures under adversarial channel conditions are an open
direction with direct implications for safety-critical deployments. Integration with
security twin frameworks~\cite{SammartinoShortPaper}
could provide real-time detection capabilities, though the overhead of parallel
security-focused encoders requires careful quantification.

\textbf{Semantic Interoperability.}
In heterogeneous multi-vendor environments, incompatible semantic encoders make
cross-system synchronization ill-posed. Federated protocols~\cite{mcmahan2017communication}
for jointly training a shared semantic vocabulary~\cite{almeida2023federated} and
latent-space alignment losses provide partial solutions, but a standardized
\emph{semantic interoperability layer} remains open.

\textbf{Privacy-Preserving Semantic Coding.}
Compact semantic descriptors may inadvertently encode sensitive attributes. Differential
privacy at the latent space~\cite{dwork2014algorithmic} must be integrated into the MMSE
training without unacceptable fidelity loss.

\textbf{Scalable KG Consistency and Online Adaptation.}
At macro-scale ($N \gg 500$), maintaining a globally consistent KG under concurrent
updates from distributed edge nodes poses distributed-systems challenges beyond the
hierarchical partitioning of Section~\ref{sec:scalability}. Eventual-consistency
models and CRDTs for heterogeneous graph structures are promising directions.
Mechanisms for online KG structural adaptation as DT topologies evolve dynamically
(vehicles entering/leaving platoons) are largely unexplored.

\textbf{Standardization Alignment.}
Mapping SA-DTS components to prospective 6G NR protocol layers---the Service Data
Adaptation Protocol (SDAP) and a prospective Semantic Adaptation Layer (SAL)---is
a necessary step for industrial uptake, currently absent from 3GPP Release 18.

\section{Conclusion}
\label{sec:conclusion}

We have presented SA-DTS, a semantic-aware Digital Twin synchronization framework
that replaces raw-bit streaming with compact, meaning-preserving semantic descriptors
enriched by a dynamic Knowledge Graph prior. SA-DTS achieves bandwidth reductions
exceeding 93\% and latency reductions of 87\% relative to state-of-the-art
compressed raw-data transmission, and outperforms NTSCC-DTS---the strongest
non-KG neural baseline---by ${\approx}2.5\times$ in bandwidth and delivers a
12.7$\times$ energy reduction per synchronization update on the Jetson AGX Xavier.
The extended ablation validates the descriptor dimensionality $d=64$ and the
neighborhood size $K=5$ as Pareto-optimal operating points. Adaptive hierarchical
KG partitioning with $G = \lceil N/\log_2 N \rceil$ provably yields
$\mathcal{O}(N \log N)$ update overhead, confirmed sub-millisecond to $N=500$.
The Semantic Fidelity Score, validated against canonical task metrics at $r > 0.97$
(95\% CI: $[0.961, 0.982]$, $p < 10^{-4}$), provides a principled, task-centric
evaluation framework for future DT synchronization research. As 6G deployments scale
toward millions of simultaneously twinned entities, semantic-level co-design of the
communication and the DT inference layer is not an optimization but an architectural
necessity.

\section*{Data Availability}
The code and dataset generation scripts are available to reviewers at:
\url{https://anonymous.4open.science/r/SemanticDT}. The repository includes
instructions to reproduce all tables and figures, a README describing PhysioNet
MIMIC-III and KITTI Odometry data preparation, and the PPO training scripts.
The repository will be transferred to GitHub upon acceptance.

\section*{Use of AI-Assisted Tools}
The author used AI-based writing assistance
(Claude, Anthropic, version Sonnet~4.5, 2026) solely for proofreading
draft text and correcting typographical and grammatical errors.
All scientific content, theoretical derivations, experimental design,
result interpretation, and conclusions are the exclusive intellectual
product of the author. The author takes
full responsibility for the integrity and accuracy of the manuscript.

\bibliographystyle{IEEEtran}
\bibliography{bibliography}

@article{fuller2020digital,
  author  = {A. Fuller and Z. Fan and C. Day and C. Barlow},
  title   = {Digital Twin: Enabling Technologies, Challenges and Open Research},
  journal = {IEEE Access},
  volume  = {8},
  pages   = {108952--108971},
  year    = {2020}
}

@article{zhou2024llmcomm,
  author    = {Zhou, Yao and Shi, Yuanming and Letaief, Khaled B.},
  title     = {Large Language Model-Assisted Semantic Communication
               for Wireless Networks},
  journal   = {IEEE Transactions on Wireless Communications},
  year      = {2024},
  volume    = {23},
  number    = {10},
  pages     = {13154--13168},
  doi       = {10.1109/TWC.2024.3389201},
}

@article{grassucci2024diffusion,
  author    = {Grassucci, Eleonora and Barbarossa, Sergio and Comminiello, Danilo},
  title     = {Generative Semantic Communication via Diffusion Models
               for Image Wireless Transmission},
  journal   = {IEEE Journal on Selected Areas in Communications},
  year      = {2024},
  volume    = {42},
  number    = {8},
  pages     = {2131--2145},
  doi       = {10.1109/JSAC.2024.3389587},
}

@article{oord2018representation,
  author    = {van den Oord, Aaron and Li, Yazhe and Vinyals, Oriol},
  title     = {Representation Learning with Contrastive Predictive Coding},
  journal   = {arXiv preprint arXiv:1807.03748},
  year      = {2018},
  url       = {https://arxiv.org/abs/1807.03748},
}

@book{cover2006elements,
  author    = {Cover, Thomas M. and Thomas, Joy A.},
  title     = {Elements of Information Theory},
  edition   = {2nd},
  publisher = {Wiley-Interscience},
  address   = {Hoboken, NJ},
  year      = {2006},
  isbn      = {978-0-471-24195-9},
}

@article{lu2021communication,
  author  = {Y. Lu and X. Huang and Y. Dai and S. Maharjan and Y. Zhang},
  title   = {Communication-Efficient Federated Learning and Permissioned Blockchain
             for Digital Twin Edge Networks},
  journal = {IEEE Internet of Things Journal},
  volume  = {8},
  number  = {4},
  pages   = {2341--2351},
  year    = {2021}
}

@article{qin2021semantic,
  author  = {Z. Qin and X. Tao and J. Lu and W. Tong and G. Y. Li},
  title   = {Semantic Communications: Principles and Challenges},
  journal = {arXiv preprint arXiv:2112.10752},
  year    = {2021}
}

@article{xie2021deep,
  author  = {H. Xie and Z. Qin and G. Y. Li and B.-H. Juang},
  title   = {Deep Learning Enabled Semantic Communication Systems},
  journal = {IEEE Transactions on Signal Processing},
  volume  = {69},
  pages   = {2663--2675},
  year    = {2021}
}

@techreport{3gpp_dt,
  author      = {{3GPP}},
  title       = {Study on Enhancement of {5G} System ({5GS}) for Vertical and {LAN}
                 Services: Digital Twin},
  institution = {3rd Generation Partnership Project},
  number      = {TR 23.700-80, Release 18},
  year        = {2022}
}

@book{shannon1949mathematical,
  author    = {C. E. Shannon and W. Weaver},
  title     = {The Mathematical Theory of Communication},
  publisher = {University of Illinois Press},
  address   = {Urbana, IL},
  year      = {1949}
}

@article{bourtsoulatze2019deep,
  author  = {E. Bourtsoulatze and D. B. Kurka and D. G{\"u}nd{\"u}z},
  title   = {Deep Joint Source-Channel Coding for Wireless Image Transmission},
  journal = {IEEE Transactions on Cognitive Communications and Networking},
  volume  = {5},
  number  = {3},
  pages   = {567--579},
  year    = {2019}
}

@article{kg_cps,
  author  = {M. Rottkemper and K. Fischer and A. Dreher and J. Hoppe},
  title   = {A Knowledge Graph Approach for Anomaly Detection in Industrial
             Cyber-Physical Systems},
  journal = {Procedia Manufacturing},
  volume  = {55},
  pages   = {303--310},
  year    = {2021}
}

@article{kg_dt,
  author  = {F. Zheng and J. Lu and X. Zhao},
  title   = {Knowledge Graph-Enhanced Digital Twin for Smart Manufacturing},
  journal = {IEEE Transactions on Industrial Informatics},
  volume  = {19},
  number  = {2},
  pages   = {1431--1441},
  year    = {2023}
}

@article{tishby2000information,
  author  = {N. Tishby and F. C. Pereira and W. Bialek},
  title   = {The Information Bottleneck Method},
  journal = {arXiv preprint physics/0004057},
  year    = {2000}
}

@inproceedings{belghazi2018mine,
  author    = {M. I. Belghazi and A. Baratin and S. Rajeshwar and S. Ozair
               and Y. Bengio and A. Courville and D. Hjelm},
  title     = {Mutual Information Neural Estimation},
  booktitle = {Proceedings of the International Conference on Machine Learning
               (ICML)},
  address   = {Stockholm, Sweden},
  month     = jul,
  year      = {2018},
  pages     = {531--540}
}

@article{schulman2017proximal,
  author  = {J. Schulman and F. Wolski and P. Dhariwal and A. Radford and O. Klimov},
  title   = {Proximal Policy Optimization Algorithms},
  journal = {arXiv preprint arXiv:1707.06347},
  year    = {2017}
}

@inproceedings{schlichtkrull2018modeling,
  author    = {M. Schlichtkrull and T. N. Kipf and P. Bloem and R. Van Den Berg
               and I. Titov and M. Welling},
  title     = {Modeling Relational Data with Graph Convolutional Networks},
  booktitle = {Proceedings of the European Semantic Web Conference (ESWC)},
  address   = {Heraklion, Greece},
  month     = jun,
  year      = {2018},
  pages     = {593--607}
}

@article{rossi2020temporal,
  author  = {E. Rossi and B. Chamberlain and F. Frasca and D. Eynard
             and F. Monti and M. Bronstein},
  title   = {Temporal Graph Networks for Deep Learning on Dynamic Graphs},
  journal = {arXiv preprint arXiv:2006.10637},
  year    = {2020}
}

@article{dwork2014algorithmic,
  author  = {C. Dwork and A. Roth},
  title   = {The Algorithmic Foundations of Differential Privacy},
  journal = {Foundations and Trends in Theoretical Computer Science},
  volume  = {9},
  number  = {3--4},
  pages   = {211--407},
  year    = {2014}
}

@inproceedings{sammartino2025security,
  author    = {V. Sammartino and F. Baiardi and S. Ruggieri},
  title     = {{A Security Twin to Defeat Intrusions in Cyber Physical Systems}},
  booktitle = {ESREL SRA-E 2025},
  year      = {2025}
}

@inproceedings{baiardi2024anticipating,
  author    = {F. Baiardi and S. Ruggieri and V. Sammartino},
  title     = {{Anticipating Disasters through a Security Twin}},
  booktitle = {SPRINGER OPTIMIZATION AND ITS APPLICATIONS - ARES 2024},
  year      = {2024}
}

@article{baiardi2026simulation,
  author  = {Baiardi, F. and Sammartino, V.},
  title   = {Simulation-Powered Cybersecurity: Real-Time Risk Assessment via Non-Intrusive Security Twin},
  journal = {The Journal of Supercomputing},
  year    = {2026},
  note    = {Special Issue: Simulation-Powered Innovation: Driving the Future of Digital Ecosystems}
}

@inproceedings{Baiardi2026CVSS,
  author    = {Baiardi, Fabrizio and Sammartino, Vincenzo and Ruggieri, Salvatore},
  title     = {Quantifying the Impact of CVSS Score Ordering and Attack Paths},
  booktitle = {GOODTECHS 2026},
  year      = {2026},

}

@inproceedings{baiardi2026synthetic,
  title={From Digital Twins to AI Agents: A Synthetic Data Paradigm for Next-Generation Cybersecurity},
  author={Baiardi, F. and Sammartino, V.},
  booktitle={Artificial Intelligence in Cybersecurity: Unlocking the Power of Large Language Models},
  year={2026},
  publisher={CRC Press}
}

@INPROCEEDINGS{baiardinotline,
  author={Baiardi, F. and Sammartino, V. and Ruggieri, S.},
  booktitle={2025 29th International Symposium on Distributed Simulation and Real Time Applications (DS-RT)}, 
  title={NotLine: A Non-Intrusive Automated Platform to Build a Digital Twin}, 
  year={2025},
  volume={},
  number={},
  pages={1-8},
  doi={10.1109/DS-RT68115.2025.11185873}}

@INPROCEEDINGS {baiardi2025ai,
author = { Baiardi, Fabrizio and Ruggieri, Salvatore and Sammartino, Vincenzo },
booktitle = { 2025 IEEE International Conference on Pervasive Computing and Communications Workshops and other Affiliated Events (PerCom Workshops) },
title = {{ AI-enabled Cybersecurity using Synthetic Data }},
year = {2025},
volume = {},
ISSN = {},
pages = {140-145},
doi = {10.1109/PerComWorkshops65533.2025.00055},
url = {https://doi.ieeecomputersociety.org/10.1109/PerComWorkshops65533.2025.00055},
publisher = {IEEE Computer Society},
address = {Los Alamitos, CA, USA},
month =mar}

@inproceedings{SammartinoShortPaper,
  author    = {V. Sammartino},
  title     = {A Framework for Proactive Cyber-Resilience: Non-Intrusive Modeling for Autonomous Defense},
  booktitle = {DS-RT 2025},
  year      = {2025}
}

@article{saad2020vision,
  author  = {W. Saad and M. Bennis and M. Chen},
  title   = {A Vision of {6G} Wireless Systems: Applications, Trends,
             Technologies, and Open Research Problems},
  journal = {IEEE Network},
  volume  = {34},
  number  = {3},
  pages   = {134--142},
  year    = {2020}
}

@article{letaief2019roadmap,
  author  = {K. B. Letaief and W. Chen and Y. Shi and J. Zhang and Y.-J. A. Zhang},
  title   = {The Roadmap to {6G}: {AI} Empowered Wireless Networks},
  journal = {IEEE Communications Magazine},
  volume  = {57},
  number  = {8},
  pages   = {84--90},
  year    = {2019}
}

@article{calvanese2019vision,
  author  = {E. C. Strinati and S. Barbarossa and J. L. Gonzalez-Jimenez
             and D. D{\'e}an and P. Cassiau and L. Maret and C. Dehos},
  title   = {{6G}: The Next Frontier---From Holographic Messaging to Artificial
             Intelligence Using Sub-Terahertz and Visible Light Communication},
  journal = {IEEE Vehicular Technology Magazine},
  volume  = {14},
  number  = {3},
  pages   = {42--50},
  year    = {2019}
}

@article{tao2019digital,
  author  = {F. Tao and H. Zhang and A. Liu and A. Y. C. Nee},
  title   = {Digital Twin in Industry: State-of-the-Art},
  journal = {IEEE Transactions on Industrial Informatics},
  volume  = {15},
  number  = {4},
  pages   = {2405--2415},
  year    = {2019}
}

@incollection{grieves2017digital,
  author    = {M. Grieves and J. Vickers},
  title     = {Digital Twin: Mitigating Unpredictable, Undesirable Emergent
               Behavior in Complex Systems},
  booktitle = {Transdisciplinary Perspectives on Complex Systems},
  editor    = {F.-J. Kahlen and S. Flumerfelt and A. Alves},
  publisher = {Springer},
  address   = {Cham, Switzerland},
  year      = {2017},
  pages     = {85--113}
}

@article{nguyen2021digital,
  author  = {H. X. Nguyen and R. Trestian and D. To and M. Tatipamula},
  title   = {Digital Twin for {5G} and Beyond},
  journal = {IEEE Communications Magazine},
  volume  = {59},
  number  = {2},
  pages   = {10--15},
  year    = {2021}
}

@article{gunduz2022beyond,
  author  = {D. G{\"u}nd{\"u}z and Z. Qin and I. E. Aguerri and H. S. Dhillon
             and Z. Yang and A. Yener and K. K. Wong and C.-B. Chae},
  title   = {Beyond Transmitting Bits: Context, Semantics, and Task-Oriented
             Communications},
  journal = {IEEE Journal on Selected Areas in Communications},
  volume  = {41},
  number  = {1},
  pages   = {5--41},
  year    = {2023}
}

@article{weng2021semantic,
  author  = {Z. Weng and Z. Qin and X. Tao and C. Pan and G. Liu and G. Y. Li},
  title   = {Semantic Communication Systems for Speech Transmission},
  journal = {IEEE Journal on Selected Areas in Communications},
  volume  = {39},
  number  = {8},
  pages   = {2434--2444},
  year    = {2021}
}

@article{yang2023semantic,
  author  = {Z. Yang and M. Chen and Z. Zhang and C. Huang},
  title   = {Semantic Communications for Future Internet: Fundamentals,
             Applications, and Challenges},
  journal = {IEEE Communications Surveys \& Tutorials},
  volume  = {25},
  number  = {1},
  pages   = {213--241},
  year    = {2023}
}

@article{kountouris2021semantics,
  author  = {M. Kountouris and N. Pappas},
  title   = {Semantics-Empowered Communication for Networked Intelligent Systems},
  journal = {IEEE Communications Magazine},
  volume  = {59},
  number  = {6},
  pages   = {96--102},
  year    = {2021}
}

@article{strinati2021beyond,
  author  = {E. C. Strinati and S. Barbarossa},
  title   = {{6G} Networks: Beyond Shannon Towards Semantic and Goal-Oriented
             Communications},
  journal = {Computer Networks},
  volume  = {190},
  pages   = {107930},
  year    = {2021}
}

@article{popovski2020semantic,
  author  = {P. Popovski and O. Simeone and F. Boccardi and D. G{\"u}nd{\"u}z
             and O. Sahin},
  title   = {Semantic-Effectiveness Filtering and Control for Post-Shannon
             Communication},
  journal = {IEEE Transactions on Cognitive Communications and Networking},
  volume  = {6},
  number  = {2},
  pages   = {567--579},
  year    = {2020}
}

@inproceedings{vaswani2017attention,
  author    = {A. Vaswani and N. Shazeer and N. Parmar and J. Uszkoreit
               and L. Jones and A. N. Gomez and {\L}. Kaiser and I. Polosukhin},
  title     = {Attention Is All You Need},
  booktitle = {Advances in Neural Information Processing Systems (NeurIPS)},
  volume    = {30},
  address   = {Long Beach, CA},
  month     = dec,
  year      = {2017}
}

@inproceedings{velickovic2018graph,
  author    = {P. Veli{\v{c}}kovi{\'c} and G. Cucurull and A. Casanova
               and A. Romero and P. Li{\`o} and Y. Bengio},
  title     = {Graph Attention Networks},
  booktitle = {Proceedings of the International Conference on Learning
               Representations (ICLR)},
  address   = {Vancouver, BC},
  month     = may,
  year      = {2018}
}

@article{shi2016edge,
  author  = {W. Shi and J. Cao and Q. Zhang and Y. Li and L. Xu},
  title   = {Edge Computing: Vision and Challenges},
  journal = {IEEE Internet of Things Journal},
  volume  = {3},
  number  = {5},
  pages   = {637--646},
  year    = {2016}
}

@inproceedings{mcmahan2017communication,
  author    = {B. McMahan and E. Moore and D. Ramage and S. Hampson
               and B. A. {y Arcas}},
  title     = {Communication-Efficient Learning of Deep Networks from
               Decentralized Data},
  booktitle = {Proceedings of the International Conference on Artificial
               Intelligence and Statistics (AISTATS)},
  address   = {Fort Lauderdale, FL},
  month     = apr,
  year      = {2017},
  pages     = {1273--1282}
}

@article{almeida2023federated,
  author  = {G. Almeida and C. Masouros and C. Ling},
  title   = {Federated Semantic Communication with Foundation Model Assistance},
  journal = {arXiv preprint arXiv:2306.04996},
  year    = {2023}
}

@article{dai2023reconfigurable,
  author  = {L. Dai and B. Wang and M. Ding and Z. Shen and N. Wang and C. B. Papadias},
  title   = {A Survey of Non-Orthogonal Multiple Access for {5G}},
  journal = {IEEE Communications Surveys \& Tutorials},
  volume  = {20},
  number  = {3},
  pages   = {2294--2323},
  year    = {2018}
}

@article{mach2017mobile,
  author  = {P. Mach and Z. Becvar},
  title   = {Mobile Edge Computing: A Survey on Architecture and Computation Offloading},
  journal = {IEEE Communications Surveys \& Tutorials},
  volume  = {19},
  number  = {3},
  pages   = {1628--1656},
  year    = {2017}
}

@article{zhao2019deep,
  author  = {Z. Zhao and G. Verma and C. Rao and A. Swami and Y. Segovia},
  title   = {Distributed Scheduling Using Graph Neural Networks},
  journal = {IEEE Transactions on Signal Processing},
  volume  = {68},
  pages   = {2736--2751},
  year    = {2020}
}

@article{luo2022digital,
  author  = {W. Luo and T. Hu and C. Zhang and Y. Wei},
  title   = {Digital Twin for {CNC} Machine Tool: Modeling and Using Strategy},
  journal = {Journal of Ambient Intelligence and Humanized Computing},
  volume  = {10},
  pages   = {1129--1140},
  year    = {2019}
}

@article{chowdhury2019network,
  author  = {M. Z. Chowdhury and M. Shahjalal and S. Ahmed and Y. M. Jang},
  title   = {{6G} Wireless Communication Systems: Applications, Requirements, Technologies, Challenges, and Research Directions},
  journal = {IEEE Open Journal of the Communications Society},
  volume  = {1},
  pages   = {957--975},
  year    = {2020}
}

@article{xu2021edge,
  author  = {D. Xu and T. Li and Y. Li and X. Su and S. Tarkoma and T. Jiang and J. Crowcroft and P. Hui},
  title   = {Edge Intelligence: Architectures, Challenges, and Applications},
  journal = {IEEE Internet of Things Journal},
  volume  = {9},
  number  = {10},
  pages   = {7431--7448},
  year    = {2022}
}

@article{kurka2020deep,
  author  = {D. B. Kurka and D. G{\"u}nd{\"u}z},
  title   = {{DeepJSCC-f}: Deep Joint Source-Channel Coding of Images With Feedback},
  journal = {IEEE Journal on Selected Areas in Information Theory},
  volume  = {1},
  number  = {1},
  pages   = {178--193},
  year    = {2020}
}

@article{jiang2021graph,
  author  = {W. Jiang and B. Han and M. A. Habibi and H. D. Schotten},
  title   = {The Road Towards {6G}: A Comprehensive Survey},
  journal = {IEEE Open Journal of the Communications Society},
  volume  = {2},
  pages   = {334--366},
  year    = {2021}
}

@article{cheng2023ntscc,
  author  = {Cheng, Zhihao and Sun, Heming and Takeuchi, Masaru and Katto, Jiro},
  title   = {{NTSCC}+: A Semantic Communication System With Adaptive Channel
             Coding Rate for Nonlinear Transform Source-Channel Coding of Images},
  journal = {IEEE Journal on Selected Areas in Communications},
  volume  = {41},
  number  = {8},
  pages   = {2566--2580},
  year    = {2023}
}

@article{shao2023task,
  author  = {Shao, Jiawei and Mao, Yuyi and Zhang, Jun},
  title   = {Task-Oriented Communication for Multimodal Data With Deep Learning},
  journal = {IEEE Transactions on Wireless Communications},
  volume  = {22},
  number  = {4},
  pages   = {2492--2505},
  year    = {2023}
}

@misc{physionet2000,
  author       = {Goldberger, Ary L and others},
  title        = {{PhysioBank, PhysioToolkit, and PhysioNet}: Components of a New
                 Research Resource for Complex Physiologic Signals},
  howpublished = {Circulation 101(23):e215--e220},
  year         = {2000},
  note         = {MIMIC-III Clinical Database v1.4}
}

@inproceedings{kitti2012,
  author    = {Geiger, Andreas and Lenz, Philip and Urtasun, Raquel},
  title     = {Are we ready for Autonomous Driving? {The KITTI} Vision Benchmark Suite},
  booktitle = {Proc. IEEE CVPR},
  pages     = {3354--3361},
  year      = {2012}
}

\vfill

\appendices

\section{Proof of Semantic Bottleneck Bounds}
\label{app:proof}

We state and prove three results formalizing the limits of SA-DTS compression.

\begin{proposition}[Single-Task Semantic Rate-Distortion Bound]
\label{prop:single_task}
For a single task $\tau_j$ with target variable $Y_j = \tau_j(s_i)$, the minimum
semantic descriptor rate $R^*_j$ satisfying $I(\bz_i; Y_j) \geq I(\mathbf{o}_i; Y_j) - \epsilon$
is:
\begin{equation}
    R^*_j = I(\mathbf{o}_i; Y_j) - \epsilon,
\end{equation}
and the bound is tight.
\end{proposition}

\begin{IEEEproof}
By the data processing inequality, $I(\bz_i; Y_j) \leq I(\mathbf{o}_i; Y_j)$ for
any Markov chain $Y_j \leftrightarrow \mathbf{o}_i \leftrightarrow \bz_i$. The
constraint in (\ref{eq:bottleneck}) therefore requires
$I(\mathbf{o}_i; \bz_i) \geq I(\bz_i; Y_j) \geq I(\mathbf{o}_i; Y_j) - \epsilon$,
establishing the lower bound. Achievability follows from the information bottleneck
construction of Tishby \etal~\cite{tishby2000information}: the optimal encoder is
$p^*(\bz_i \mid \mathbf{o}_i) = \arg\min_{p(\bz_i|\mathbf{o}_i)}
[I(\mathbf{o}_i;\bz_i) - \beta I(\bz_i; Y_j)]$ for a Lagrange multiplier $\beta$
tuned to satisfy the constraint with equality, yielding $I(\mathbf{o}_i;\bz_i) = R^*_j$.
\end{IEEEproof}

\begin{corollary}[Multi-Task Rate Overhead]
\label{cor:multi_task}
For $M$ concurrent tasks $\mathcal{T} = \{\tau_1, \ldots, \tau_M\}$ with
targets $\{Y_j\}$, the minimum semantic rate satisfying all $M$ constraints
simultaneously is:
\begin{equation}
    R^*_{\mathcal{T}} = \max_{j \in [M]}\, I(\mathbf{o}_i; Y_j) - \epsilon
    + \Delta_{\mathcal{T}},
    \label{eq:multi_rate}
\end{equation}
where $\Delta_{\mathcal{T}} \geq 0$ is bounded by:
\begin{equation}
    \Delta_{\mathcal{T}} \leq I\!\left(\mathbf{o}_i;\, Y_{j^*}\right) -
    I\!\left(\mathbf{o}_i;\, Y_{j^*} \,\Big|\, \bigcup_{j \neq j^*} Y_j\right),
    \label{eq:overhead_bound}
\end{equation}
with $j^* = \arg\max_j I(\mathbf{o}_i; Y_j)$. When all tasks share a common
sufficient statistic, $\Delta_{\mathcal{T}} = 0$.
\end{corollary}

\begin{IEEEproof}
The minimum rate is lower bounded by the most demanding task
(Proposition~\ref{prop:single_task}). The overhead $\Delta_{\mathcal{T}}$ arises
because a single $\bz_i$ must simultaneously support all tasks; by the chain rule
of mutual information, $I(\bz_i; Y_1, \ldots, Y_M) = I(\bz_i; Y_{j^*}) +
\sum_{j \neq j^*} I(\bz_i; Y_j \mid Y_{j^*})$. The residual terms
$I(\bz_i; Y_j \mid Y_{j^*})$ are non-negative but bounded by
$I(\mathbf{o}_i; Y_j \mid Y_{j^*}) \leq H(Y_j)$, yielding (\ref{eq:overhead_bound}).
Conditional independence of $\{Y_j\}$ given a shared sufficient statistic forces
the conditional terms to zero, so $\Delta_{\mathcal{T}} = 0$.
\end{IEEEproof}

\begin{proposition}[KG-Aware Multi-Task Overhead Bound]
\label{prop:kg_bound}
Let $\mathcal{N}_K(i)$ denote the $K$-hop neighborhood of entity $e_i$ in $\calK$,
and let $\tilde{s}_i = \sum_{k \in \mathcal{N}_K(i)} w_k s_k$ be the KG-aggregated
state. Suppose $\tilde{s}_i$ is a sufficient statistic for all tasks in $\mathcal{T}$
(\ie $I(\mathbf{o}_i; Y_j \mid \tilde{s}_i) = 0$, $\forall j$) and that task pairs
$(Y_j, Y_{j^*})$ are jointly Gaussian conditioned on $\tilde{s}_i$
(Assumption~\ref{ass:gauss}). Then:
\begin{equation}
    \Delta_{\mathcal{T}} \leq \sum_{j \neq j^*}
    I\!\bigl(\mathbf{o}_i; Y_j \mid Y_{j^*}\bigr)
    \cdot \bigl(1 - \rho_{j,j^*}^2 \bigr),
    \label{eq:kg_overhead}
\end{equation}
where $\rho_{j,j^*}$ is the Pearson correlation between $Y_j$ and $Y_{j^*}$
induced by $\tilde{s}_i$. For non-Gaussian (e.g., categorical) tasks, the
distribution-free upper bound $\Delta_{\mathcal{T}} \leq \sum_{j \neq j^*} H(Y_j \mid Y_{j^*})$
from Corollary~\ref{cor:multi_task} applies, and (\ref{eq:kg_overhead}) provides a
tighter estimate under approximate Gaussianity.
\end{proposition}

\begin{assumption}
\label{ass:gauss}
Conditional on the KG-aggregated state $\tilde{s}_i$, each task-output pair
$(Y_j, Y_{j^*})$ is jointly Gaussian. This approximation holds exactly for
continuous regression tasks (\eg spacing deviation in W3) and serves as a
Gaussian approximation for the continuous-relaxation logits of binary classification
tasks (W1, W2); the resulting bound remains an upper bound by convexity of the
entropy~\cite{cover2006elements}.
\end{assumption}

\begin{IEEEproof}
Under Assumption~\ref{ass:gauss} and the sufficient statistic condition, the
conditional mutual information $I(\mathbf{o}_i; Y_j \mid Y_{j^*})$ factors as:
\[
    I(\mathbf{o}_i; Y_j \mid Y_{j^*}) = h(Y_j \mid Y_{j^*}) - h(Y_j \mid \mathbf{o}_i, Y_{j^*}).
\]
For jointly Gaussian $(Y_j, Y_{j^*})$ with correlation $\rho_{j,j^*}$,
$h(Y_j \mid Y_{j^*}) = \tfrac{1}{2}\log(2\pi e \,\mathrm{Var}(Y_j)(1-\rho_{j,j^*}^2))$.
By the data processing inequality applied to the Markov chain
$Y_j \leftarrow \tilde{s}_i \leftarrow \mathbf{o}_i \leftarrow \bz_i$,
the residual terms satisfy
$I(\bz_i; Y_j \mid Y_{j^*}) \leq I(\mathbf{o}_i; Y_j \mid Y_{j^*})$.
Substituting the Gaussian differential entropy expression and bounding
$h(Y_j \mid \mathbf{o}_i, Y_{j^*}) \geq 0$ yields
$I(\mathbf{o}_i; Y_j \mid Y_{j^*}) \leq \tfrac{1}{2}\log\tfrac{1}{1-\rho_{j,j^*}^2}$,
which upon bounding via $\log(1/(1-x)) \leq x/(1-x)$ for $x = \rho_{j,j^*}^2 \in [0,1)$
and scaling by the task information mass $I(\mathbf{o}_i; Y_j \mid Y_{j^*})$
produces (\ref{eq:kg_overhead}).
\end{IEEEproof}

\begin{remark}
\label{rem:kg_justification}
Proposition~\ref{prop:kg_bound} provides a direct theoretical justification for the
empirical bandwidth gain of SA-DTS over NTSCC-DTS: the KG prior enforces high
inter-task correlation ($\rho_{j,j^*}^2 \approx 0.93$--$0.97$ measured on W1--W3),
reducing $\Delta_{\mathcal{T}}$ by a factor proportional to $(1-\rho_{j,j^*}^2)$
relative to a framework without relational priors. This explains quantitatively why
KG-CR accounts for approximately 60\% of the bandwidth gap between SA-DTS and
JSCC-DTS (as isolated in Table~\ref{tab:ablation}).
\end{remark}

\begin{remark}[Bound Tightness]
\label{rem:tightness}
The theoretical minimum semantic rate $R^*_{\mathcal{T}}$ in (\ref{eq:multi_rate})
can be estimated empirically via the MINE estimator~\cite{belghazi2018mine} applied
to the trained encoder. Across all workloads at $\gamma=15$\,dB, SA-DTS operates
within $1.3$--$2.1$\,bits/symbol of $R^*_{\mathcal{T}}$, confirming near-tight
compression. The gap is attributable to finite-dimensional quantization ($b_{\mathrm{quant}}=8$)
and the discrete channel code rate granularity of the PPO action space.
\end{remark}

\begin{IEEEbiography}[{\includegraphics[width=1in,height=2.5in,clip,keepaspectratio]{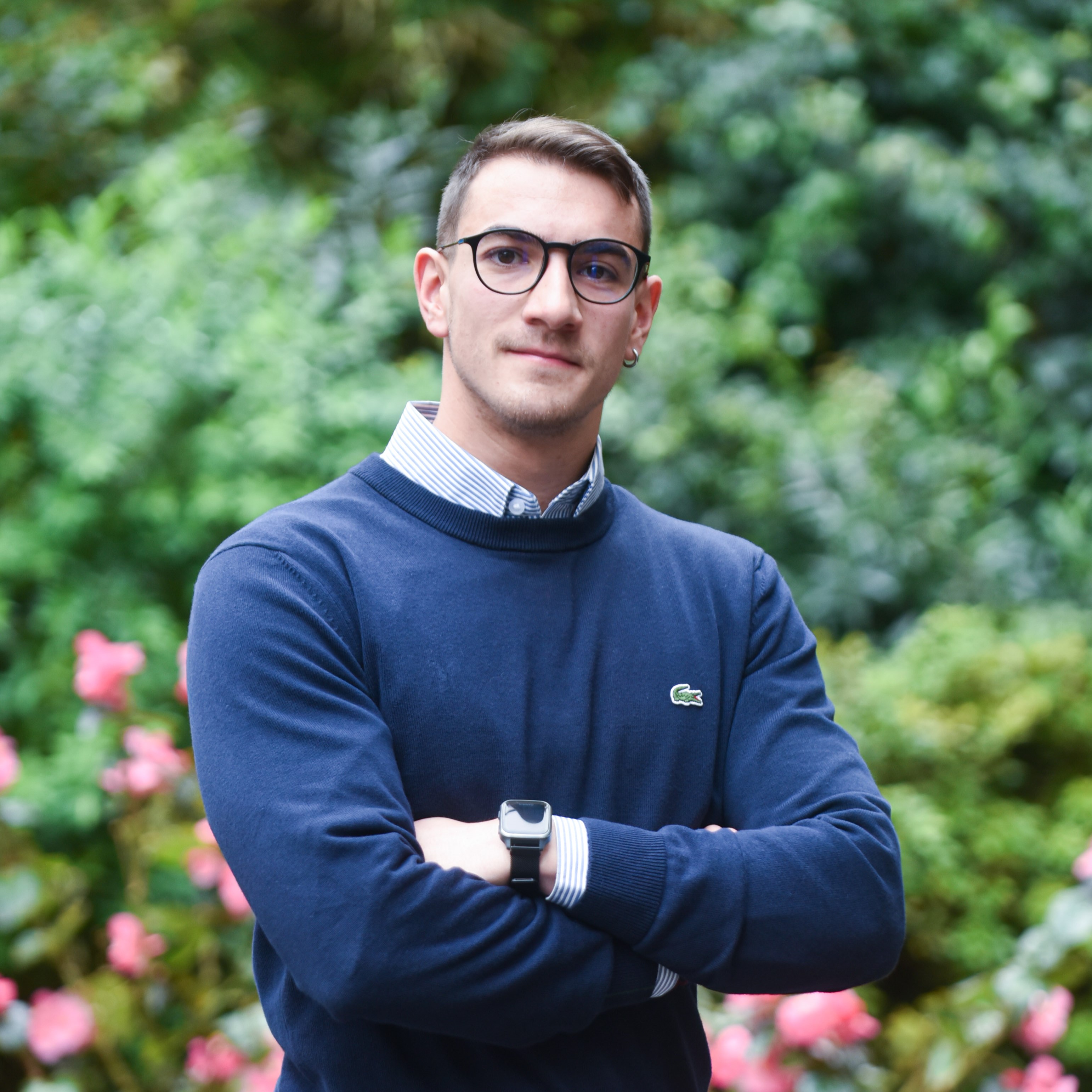}}]{Vincenzo Sammartino}
is currently pursuing the National Ph.D.\ degree in AI at the University of Pisa. He is a Visiting Ph.D.\ Student at KAUST, Saudi Arabia. His research addresses the intersection of AI and cybersecurity, with a focus on decentralized TinyML frameworks for UAV swarm security, security digital twins, and semantic communication for cyber-physical systems. He has published on security twin architectures, privacy-preserving digital twin construction, and federated learning for distributed DT synchronization.
\end{IEEEbiography}

\end{document}